\documentclass[runningheads]{llncs}
\usepackage[T1]{fontenc}
\usepackage{graphicx}
\usepackage{amsmath}
\usepackage{amssymb}
\usepackage{amsfonts}
\usepackage{proof}
\usepackage{tikz-cd}
\usepackage{multicol}
\usepackage{multirow}
\usepackage{enumitem}
\usepackage{subcaption}
\usepackage[normalem]{ulem}
\usepackage{mathrsfs}
\usepackage{url}
\usepackage{color}
\usepackage[linktoc=all,hidelinks,colorlinks=true,linkcolor=blue,citecolor=blue]{hyperref}
\usepackage{cleveref}
\usetikzlibrary{positioning}
\usepackage{todonotes}
\usepackage{orcidlink}
\usepackage{fullpage}

\usepackage{mycommands}

\begin{document}

\title{The Modal Cube Revisited: Semantics without Worlds\\(Technical Report)}

\author{Renato Leme\inst{1}\orcidlink{0000-0001-8199-1261} 
\and Carlos Olarte\inst{2}\orcidlink{0000-0002-7264-7773} 
\and 
Elaine Pimentel\inst{3}\orcidlink{0000-0002-7113-0801}
\and Marcelo E. Coniglio\inst{1}\orcidlink{0000-0002-1807-0520}}

\authorrunning{R. Leme et al.}

\institute{Centre for Logic, Epistemology and The History of Science, UNICAMP, Brazil\\
\email{rntreisleme@gmail.com}, \email{coniglio@unicamp.br} 
\and Universit\'{e} Sorbonne Paris Nord, CNRS, Laboratoire d'Informatique de Paris Nord, Villetaneuse, France
 \\ \email{olarte@lipn.univ-paris13.fr}
\and Department of Computer Science, University College London, UK\\ \email{e.pimentel@ucl.ac.uk}
}

\maketitle

\begin{abstract}
We present a non-deterministic semantic framework for all modal logics in the
modal cube, extending prior works by Kearns and others. Our approach introduces
modular and uniform multi-valued non-deterministic matrices (Nmatrices) for
each logic, where necessitation is captured by the systematic use of level
valuations. The semantics is grounded in an eight-valued system and provides a
sound and complete decision procedure for each modal logic, extending and
refining earlier semantics as particular cases. Additionally, we propose a
novel model-theoretic perspective that links our framework to relational
(Kripke-style) semantics, addressing longstanding questions regarding the
correspondence between modal axioms and semantic conditions in
non-deterministic settings. 
This yields a philosophically robust and technically modular alternative to traditional possible-world
semantics.

\keywords{Modal logics semantics \and Non deterministic matrices}

 \end{abstract}

\section{Introduction}\label{sec:intro}

Modal logics are built on top of propositional classical logic ($\PCL$) by introducing the modal operators $\Box$ and $\DIA$, which, under the alethic interpretation, correspond to {\em necessary} and {\em possible}, respectively. 
{Thus from a propositional formula $\alpha$, which in $\PCL$ has value either false ($\vF$) or true ($\vT$), one can construct the statements {\em necessarily} $\alpha$ and {\em possibly} $\alpha$, denoted by $\Box \alpha$ and $\DIA \alpha$, respectively.}

When interpreting these formulas in Kripke semantics, modalities {\em qualify} the notion of truth, which now depends on the world in which  $\alpha$ is being evaluated. 
 Kripke's concept of {\em possible worlds} thus provides an elegant framework for capturing the possible values of propositions across different worlds. \noindent
This natural generalization of the straightforward {truth table} semantics of $\PCL$ makes it an ideal foundation for the semantics of modal logics.

{A natural question that arises is whether one could define a single finite-valued truth table capable of capturing modal logics. The answer is negative. As early as the 1930s, G\"{o}del~\cite{feferman86goedel} proved that intuitionistic logic admits no finite-valued truth-functional semantics and, since it can be faithfully embedded in the modal logic~$\mSfour$, this implies that~$\mSfour$ itself is not finite-valued. In 1940, Dugundji extended this result to the entire modal cube~\cite{DBLP:journals/jsyml/Dugundji40} and this largely halted the study of truth tables for modal logics}\footnote{{Most recently, in~\cite{DBLP:journals/logcom/Gratz22} Gr\"{a}tz sealed the fate of this line of inquiry by showing that even non-deterministic finite truth tables are insufficient.}}.

In an effort to circumvent Dugundji's result--and as a compelling alternative to Kripke semantics--Kearns~\cite{DBLP:journals/jsyml/Kearns81} and Ivlev~\cite{Ivlev1988-IVLASF} independently developed frameworks for characterizing certain modal systems\footnote{Ivlev proposed semantics for a range of non-normal modal systems lacking the necessitation rule--representing weaker versions of $\mKT$ and $\mSfive$ (later extended in, \eg,~\cite{DBLP:journals/igpl/ConiglioCN20,DBLP:journals/jphil/PawlowskiS24,DBLP:journals/logcom/PawlowskiS25}). In this work, we will focus on the extension of Kearn's work to the (normal) modal cube. A related approach can be found in~\cite{DBLP:journals/corr/abs-2501-00492} and~\cite{ConiglioPS}.}.
Kearns used four-valued multivalued truth-functions\footnote{The notion of multivalued truth-functions was later formalized by Avron and Lev under the term {\em Nmatrices}~\cite{DBLP:journals/logcom/AvronL05}. See also~\cite{Avron2011}.}, and introduced a constraint on valuations known as {\em level-valuations} to characterize the modal systems $\mKT$, $\mSfour$, and $\mSfive$. The idea is that the level-0 corresponds to the standard distribution of truth values (just following the corresponding matrix), while  higher levels filter out valuations that violate the necessitation rule, preserving only those that assign the designated truth value to tautologies. 
Validity is therefore defined as the (infinite) intersection of all level valuations, corresponding to the closure w.r.t. necessitation. This layered structure gives to necessity a {\em global interpretation}.

Kearns' (and Ivlev's) contributions remained largely overlooked and relatively obscure until they were independently revisited in two separate lines of research. In~\cite{DBLP:journals/jancl/ConiglioCP15}, Coniglio, Fari\~{n}as del Cerro, and Peron reconstructed and extended these earlier results, proposing a characterization of the systems $\mKB$ using four-valued Nmatrices with level-valuations. They also introduced six-valued Nmatrices for a range of modal systems, including $\mKT$, $\mSfour$, $\mSfive$, $\mKD$, $\mKDB$, $\mKDfour$, and $\mKDfourfive$. Around the same time, Omori and Skurt, in~\cite{DBLP:journals/flap/SkurtO16}, also began by revisiting Kearns-Ivlev's original framework, and proposed an extension using eight truth values capable of capturing the modal logic $\mK$ and six truth values for $\mKTB$. While there is a considerable conceptual overlap between~\cite{DBLP:journals/jancl/ConiglioCP15} and~\cite{DBLP:journals/flap/SkurtO16}, the two works were developed independently.

Although conceptually interesting, these works have limited practical applicability, as determining level-valuations requires accounting for the valuations of {\em all} tautologies across {\em all} levels. This requires a twofold infinite testing: on formulas {\em and} levels.

This situation remained unchanged until the work of Gr\"{a}tz~\cite{DBLP:journals/logcom/Gratz22}, whose key contribution was the introduction of a \emph{decision procedure} for Nmatrices that is both sound and complete w.r.t. Kearns' semantics. This breakthrough renewed the attention to an otherwise underexplored area, 
ultimately leading to the extension of the methods of Kearns and Gr\"{a}tz to a broader class of logics.
{For instance,~\cite{IPL} presents a semantic characterization of propositional intuitionistic logic (IPL) using a three-valued non-deterministic matrix with a restricted set of valuations, enabling a remarkably simple decision procedure for IPL.}

This work follows this path by providing a decision procedure based on multi-valued non-deterministic matrices for {\em all logics} in the modal cube. In particular, our approach focuses on the following key aspects.

\noindent
{\bf Modularity.} The results
in~\cite{DBLP:journals/jancl/ConiglioCP15,DBLP:journals/flap/SkurtO16,DBLP:journals/igpl/ConiglioCN20}
share a common limitation: a lack of modularity. These works aim to ``explain''
or ``refine'' Kearns' original approach, not addressing the fundamental
challenge of uniformly extending it. In fact, different systems use different
sets of truth-values, often diverging significantly from standard choices. An
alternative axiomatization was proposed by Pawlowski and~la
Rosa~\cite{DBLP:journals/logcom/PawlowskiR22}, resulting in a modular rule of
necessitation to $\mKT$, $\mKTB$, $\mSfour$, and $\mSfive$ (and some non-normal
modal logics).

In contrast, our work adopts a fundamentally different approach: we begin (\Cref{sec:level}) with a uniform set of eight truth values and systematically develop level-valuations for {\em all} logics in the modal cube. We demonstrate that, under certain modal axioms, some of these truth values are eliminated, thereby recovering many of the semantics proposed in the {\em op. cit.} works.
Our definitions are guided by a \emph{modal characterization} of the truth values, which enables modular procedures for proving soundness and completeness of the semantics, thus unifying and generalizing existing systems in the literature. Notably, this generalization allowed for new level-semantics
for $\mKfour, \mKfive, \mKfourfive, \mKDfour, \mKDfive$ and $\mKBfive$. 

\noindent {\bf Decision procedure.} We propose new Nmatrix-based decision procedures for all 15 normal modal logics in the modal cube (\Cref{sec:truth}). This is achieved uniformly through the modal characterization of truth values. In doing so, we extend Gr\"{a}tz's work on $\mKT$ and $\mSfour$ to the entire cube, thereby completing the picture for this family of logics. These results highlight the potential of this alternative semantic framework.

\noindent
{\bf Modal Semantics without Possible Worlds?} Kearns concludes his paper with the following striking statement~\cite[p. 86]{DBLP:journals/jsyml/Kearns81}:
\begin{quote}
``The present semantic account [\ldots] is simpler than the standard account in virtue of having dispensed with possible worlds and their relations. I also think that my account is philosophically preferable to the standard account for having done this. For I do not think there are such things as possible worlds, or even that they constitute a useful fiction.''
\end{quote}

Unfortunately, the price of rejecting possible worlds may seem steep: one must contend with multi-valued truth values, non-deterministic matrices, and--prior to the development of decision procedures--at least two levels of valuations.
As a final contribution, we reestablish the connection with Kripke semantics by linking matrix filters to Kripke models,
thus providing 
 a more ``ecumenical'' perspective where the two semantics can coexist. 
 We thus settle a longstanding 
reflection posed by Omori and Skurt~\cite[p.
27]{DBLP:journals/flap/SkurtO16}: 
\begin{quote}
``One of the virtues of Kripkean semantics is the correspondence between axioms and the accessibility relations of the Kripke frame [\ldots] But a glance at the [Kearnsean] semantics for the systems [\ldots] introduced here reveals that {\em if there is a correspondence it is not a simple one}.''
\end{quote}

\section{Preliminaries}\label{sec:mK}

In what follows, $\For$ is the set of well-formed formulas in a propositional classical
language, $\wp(\mathcal{\For})$ is the powerset of $\mathcal{\For}$, 
$\alpha,\beta$ (resp. $\Delta,\Gamma,\Lambda$) range over 
elements of $\For$ (resp. $\wp(\mathcal{\For})$), and 
$\Lan$ is a Tarskian propositional
logic, that is, a pair $(\For , \cL)$ where $\cL$ is a consistent Tarskian
consequence relation. We say that $\Delta$ is a {\em consistent set} in $\Lan$
if there is no formula $\alpha$ such that $\Delta \cL \alpha$ and  $\Delta \cL
\neg \alpha$. $\Delta$ is {\em maximally consistent} if it is consistent and
$\Delta \cup \{ \beta \}$ is inconsistent for all
$\beta\in\For\backslash\Delta$. A maximally consistent set $\Delta$ is {\em
$\alpha$-saturated} iff $\Delta \nvdash^{\Lan} \alpha$. 
For any Tarskian logic $\Lan$, if $\Gamma
\not\cL\alpha$, then there is some $\alpha$-saturated set $\Delta$ 
such that $\Gamma \subseteq \Delta$ (see \eg\ \cite{Wojcicki1984-WJCLOP}).

The principle of compositionality states that the truth-value of a 
formula is fully determined by the truth-values of its subformulas. In the presence of 
incomplete or uncertain information, this principle is relaxed
using non-deterministic matrices
(Nmatrices)~\cite{DBLP:journals/logcom/AvronL05,Avron2011}, which allows the
truth-value of a formula to be selected from a set of possible values (rather 
than from a single one). 
\begin{definition}[Nmatrix]\label{def:nmatrix}
An {\em  Nmatrix} for $\Lan$ is a tuple $\M = \langle \V,\D,\Om\rangle$, where:
\begin{itemize}
\item $\V$ is a non-empty set of truth values.
\item $\D$ (designated truth values) is a non-empty proper subset of $\V$.
\item For every n-ary connective $\conn$ in $\Lan$, $\Om$ includes a  non-deterministic truth-function $\tilde\conn: \V^n \to \wp(\V)\backslash\varnothing$.
\end{itemize}
\end{definition}

\begin{definition}[Valuation]\label{def:val}
Let $\M = \langle \V,\D,\Om\rangle$ be an Nmatrix and $\Lambda\subseteq\For$  closed under subformulas.
A {\em partial valuation} in $\M$ is a function $\vv:\Lambda\to \V$ such that, for each n-ary connective $\conn$ in  $\Lan$, the following holds for all $\alpha_0, \ldots, \alpha_n\in\Lambda$:
$
   v (\conn (\alpha_0,\ldots, \alpha_n)) \in \tilde\conn (v(\alpha_0), \ldots, v(\alpha_n))
   $.
A partial valuation in $\M$ is a (total) {\em valuation} if its domain is $\For$. 
We denote by $\Val{\M}$ the set $\{v: \For \to \V \mid v \mbox{ is a valuation in } \M\}$.
\end{definition}

\noindent
The syntax of the modal logics considered here is given by the grammar
\[\alpha ::= p \mid  \bot \mid \alpha \to \alpha \mid \Box \alpha\]
where $p\in\At$, the set of propositional variable symbols. The other usual connectives
 for possibility $\DIA$, conjunction $\wedge$, disjunction $\vee$ and
negation $\neg$ are defined as abbreviations, \eg\  $\neg \alpha$ and $\DIA
\alpha$ abbreviate $\alpha\to\bot$ and $\neg\Box\neg \alpha$, respectively. 

The basic modal logic $\mK$  
is obtained by extending the ordinary Hilbert axioms for propositional classical logic with the modal axiom $\mk$, along with the rules of  necessitation and modus ponens:

$\mk: \Box(\alpha \to \beta) \to (\Box\alpha \to \Box\beta)\qquad
\vcenter{\infer[\mpo]{\beta}{\alpha\to\beta& \alpha}}\qquad\vcenter{\infer[\nec]{\Box\alpha}{\alpha}}
$

The {\em modal cube} is formed by extending $\mK$ with any
non-redundant combination of axioms $\md, \mt, \mb, \mfour$ and  $\mfive$. Such 
axioms characterize the usual frame conditions in the relational semantics. 
\begin{center}
\begin{minipage}[c]{0.55\textwidth}
$
\begin{array}{lll}
\md: \Box \alpha  \to \DIA\alpha &\quad & \mbox{Seriality}\\
\mt: \Box \alpha  \to \alpha & & \mbox{Reflexivity} \\
\mb: \alpha  \to \Box\DIA\alpha & & \mbox{Symmetry}\\
\mfour: \Box \alpha  \to \Box\Box\alpha & & \mbox{Transitivity} \\
\mfive: \DIA\alpha  \to \Box\DIA\alpha & & \mbox{Euclidianness}
\end{array}
$
\end{minipage}
\begin{minipage}[c]{0.22\textwidth}
\hbox to .4\textwidth{\hss\scalebox{.40}{
    \begin{tikzpicture}
      [every node/.style={inner sep=1pt,outer sep=0},
      logic/.style={shape=circle,draw}]
      \def\xx{4.5}\def\hxx{2}\def\hhxx{2}
      \def\yy{4}\def\hyy{2}\def\hhyy{1}
      \def\zz{-4}\def\hzz{-2}\def\hhzz{-1}
      \node[logic,label=225:\proofsystem{K}]   (ik)   at (0, 0, 0)            {} ;
      \node[logic,label=-45:\proofsystem{KB}]  (ikb)  at (\xx, 0, 0)          {} ;
      \node[logic,label=-45:\proofsystem{KB5}] (ikb5) at (\xx, 0, \zz)        {} ;
      \node[logic,label=170:\proofsystem{K4}]  (ik4)  at (0, 0, \zz)          {} ;
      \node[logic,label=180:\proofsystem{D}]   (id)   at (0, \hyy, 0)         {} ;
      \node[logic,label=135:\proofsystem{T}]   (it)   at (0, \yy, 0)          {} ;
      \node[logic,label=135:\proofsystem{S4}]  (is4)  at (0, \yy, \zz)        {} ;
      \node[logic,label=45:\proofsystem{S5}]   (is5)  at (\xx, \yy, \zz)      {} ;
      \node[logic,label=-45:\proofsystem{TB}]  (itb)  at (\xx, \yy, 0)        {} ;
      \node[logic,label=0:\proofsystem{DB}]    (idb)  at (\xx, \hyy, 0)       {} ;
      \node[logic,label=135:\proofsystem{D4}]  (id4)  at (0, \hyy, \zz)       {} ;
      \node[logic,label=-45:\proofsystem{D45}] (id45) at (\hxx, \hyy, \zz)    {} ;
      \node[logic,label=-45:\proofsystem{K45}] (ik45) at (\hxx, 0, \zz)       {} ;
      \node[logic,label=-10:\proofsystem{D5}]  (id5)  at (\hhxx, \hyy, \hhzz) {} ;
      \node[logic,label=-10:\proofsystem{K5}]  (ik5)  at (\hhxx, 0, \hhzz)    {} ;
      \draw[line width=.6pt,color=black!40]
      (ik) -- (ik4) -- (id4) (ik4) -- (ik45) ;
      \draw[line width=.7pt,color=black!60]
      (ik45) -- (ikb5)
      (id) -- (id4) -- (is4) (id4) -- (id45) -- (id5) -- (id)
      (ik) -- (ik5) -- (ik45) -- (id45) (ik5) -- (id5)
      (id45) to[bend left] (is5) ;
      \draw[line width=.8pt]
      (ik) -- (ikb) -- (ikb5) -- (is5) -- (is4) -- (it) -- (id) -- (ik)
      (it) -- (itb) -- (is5)
      (itb) -- (idb) -- (ikb)
      (id) -- (idb) ;
    \end{tikzpicture} }\hss}
\end{minipage}
\end{center}

\begin{notation}\label{rem:logics}
    We use the  shorthands below to identify some 
    sets of logics: 
\\

    \begin{tabular}{lll l lll}
        $\mK\star$  &=&    $\{\mK,  \mKB, \mKfour, \mKfive, \mKfourfive\}$ &\qquad&
        $\mKD\star$ &=&    $\{\mKD, \mKDB, \mKDfour, \mKDfive, \mKDfourfive\}$\\
        $\mKT\star$ &=&    $\{\mKT, \mKTB, \mKTfour, \mKTBfourfive\}$& \qquad & 
        $\mKBfourfive$ &=& $\{\mKBfourfive\}$
    \end{tabular}
\\

  \noindent Observe that $\mKTfour=\mSfour$, $\mKTBfourfive=\mSfive$ and $\mKBfourfive=\mKBfive$.
  In the forthcoming sections,   $\Lan$ always ranges over one of these fifteen normal modal logics. 
\end{notation}

\section{Level Valuations for the Modal Cube}\label{sec:level}
As already mentioned in the Introduction, Nmatrices do not capture the behavior of the necessitation rule. Kearns addressed this issue by restricting the set of valuations using {\em levels}: if a formula $\alpha$ receives a designated value with respect to all possible valuations--that is, if it is a tautology at a certain level--then $\Box\alpha$ will also receive a designated value, being a tautology at the next level. In other words, all valuations $v$ such that $v(\Box\alpha) \notin \D$, and $\alpha$ is a tautology,  are eliminated from the set of acceptable valuations.

This section  introduces the notion of level-valuation semantics for all
the logics in the modal cube in a modular way. We begin by presenting the
8-valued non-deterministic semantics for the logic $\mK$ which is inspired by, but different from, the version presented
in~\cite{DBLP:journals/flap/SkurtO16} (\Cref{sec:values}). Truth values are
given a \emph{modal characterization}, and we show that some of these values are 
eliminated in the semantics of certain extensions of $\mK$. We then construct
Nmatrices for all the logics in the modal cube and prove that
the semantics is sound with respect to each logic (\Cref{sec:matrices}).
Finally, we establish completeness by systematically defining
\emph{characteristic functions} directly from the meanings of the values
(\Cref{sec:comp}).

\subsection{About Truth Values and Their Meaning}\label{sec:values}

There are eight ways to ``qualify'' the truth of a formula~$\alpha$: it can be \{true, necessary, possible\} or not. These combinations are shown in~\Cref{table:truthVals}, giving rise to the eight truth-values considered in this work: $\V = \{\vF, \vf, \vff, \vfff, \vttt, \vtt, \vt, \vT\}$.

For example, if a formula~$\alpha$ has value~$\vF$, this means that~$\alpha$ is impossible: it is neither valid, nor necessary, nor possible. In other words, not only $\alpha$ does not hold ($\neg \alpha$), but its negation is both possible ($\Diamond \neg \alpha$) and necessary ($\Box \neg \alpha$). 

This intuitive interpretation will be made precise in~\Cref{sec:comp}, but for now, we proceed by developing the concepts guided by this intuition.

\begin{definition}[Value func.]\label{def:func}
Let $\val \in \V$. The {\em value function} $\val:\For\to \For$ 
 maps each $\alpha \in \For$ to the formula in the second column of \Cref{table:truthVals}.
\end{definition}

\begin{table}[!t]
    \begin{subtable}[b]{.45\textwidth}
        \centering
        \resizebox{.83\textwidth}{!}{
    \begin{tabular}{|l|c|}
    \hline
    \textbf{Truth-value}  & \textbf{Intuitive meaning}                       \\ \hline
     \; $v (\alpha) = \vF$   & $\Diamond \neg \alpha\wedge \neg \alpha\wedge \Box \neg \alpha$ \\ \hline
    \;  $v (\alpha) = \vf$   & $\Diamond \neg \alpha\wedge \neg \alpha\wedge \Diamond \alpha$  \\ \hline
   \;   $v (\alpha) = \vff$  & $\Box \alpha\wedge \neg \alpha\wedge \Box \neg \alpha$          \\ \hline
    \;  $v (\alpha) = \vfff$ & $\Box \alpha\wedge \neg \alpha\wedge \Diamond \alpha$           \\ \hline
   \;   $v (\alpha) = \vttt$ & $\Diamond \neg \alpha\wedge \alpha\wedge \Box \neg \alpha$      \\ \hline
    \;  $v (\alpha) = \vtt$  & $\Box \alpha\wedge \alpha\wedge \Box \neg \alpha$               \\ \hline
  \;    $v (\alpha) = \vt$   & $\Diamond \neg \alpha\wedge \alpha\wedge \Diamond \alpha$       \\ \hline
  \;    $v (\alpha) = \vT$   & $\Box \alpha\wedge \alpha\wedge \Diamond \alpha$                \\ \hline 
    \end{tabular}
}
    \caption{Values and  \emph{modal characterization}}\label{table:truthVals}
	\end{subtable}
    \begin{subtable}[b]{.33\textwidth}
        \centering
        \resizebox{1.3\textwidth}{!}{
    \begin{tabular}{l|l}
    \textbf{Distinguished sets}  & \textbf{\,Main feature}  \\ \hline
    \; $\mathcal{D} = \{ \vT, \vt, \vtt, \vttt \}$ &\quad $\alpha$ is true\\
    \; $\mathcal{D}^{\complement} = \{ \vF, \vf, \vff, \vfff \}$ &\quad $\neg \alpha$ is true\\
 \;  $\sN = \{ \vT, \vtt, \vfff, \vff \}$  &\quad  $\alpha$ is necessary\\
 \; $\sI = \{ \vF, \vff, \vttt, \vtt \}$  &\quad  $\neg \alpha$ is necessary\\
 \; $\sP = \{ \vT, \vt, \vfff, \vf \}$ &\quad  $\alpha$ is possible\\
 \; $\sPN = \{ \vF, \vf, \vttt, \vt \}$ &\quad  $\neg \alpha$ is possible 
    \end{tabular}
}
    \caption{Distinguished sets}\label{table:ds}
	\end{subtable}
\caption{Truth-values and distinguished sets}\label{tab:values}
\vspace{-0.6cm}
    \end{table}

\Cref{table:ds} classifies the truth-values
according to the necessity or possibility of~$\alpha$
or~$\neg\alpha$. For instance, $\vt(\alpha)=\Diamond \neg \alpha\wedge \alpha\wedge \Diamond \alpha$
and hence this value is in the sets $\D$ (designated), $\sP$ ($\alpha$ is possible)
and $\sPN$ ($\neg \alpha$ is possible). 
It is worth noticing  that the values $\vtt$ and $\vff$ are, at the same
time,  necessary ($\sN$) and impossible ($\sI$). As we will see in Section~\ref{sec:rel}, these represent, in our relational model, valuations with no successor states, in which both $\Box \alpha$ and $\Box \neg\alpha$ hold trivially.

Next, we show that the intuitive meaning of values is well defined: a
consistent set cannot prove both $\val(\alpha)$ and $\valtwo(\alpha)$ for 
two different truth-values $\val$ and $\valtwo$. 
{Moreover, we show that the truth values ${\vtt, \vff}$ are {\em stable} in the sense that, if they appear in a valuation, then the entire valuation (\ie, the corresponding row in the Nmatrix) is composed solely of these values.}

\begin{lemma}[Consistency on Values]\label{lemma:cons:vals}
Let $\Lan$ be a modal logic, 
$\val$ and $\valtwo$ be two different truth-values, $\alpha,\beta$ be formulas,  and $\Delta$ be a
consistent set in $\Lan$. 
Then
\begin{enumerate}
\item if $\Delta \vdash^{\Lan}
\funchar{\val}{\alpha}$ then $\Delta \not\vdash^{\Lan} \funchar{\valtwo}{\alpha}$; and
\item if $\val\in\{\vtt, \vff\}$, $\Delta \vdash^{\Lan} \funchar{\val}{\alpha}$ and 
 $\Delta \vdash^{\Lan} \funchar{\valtwo}{\beta}$ then $\valtwo\in \{\vtt,\vff\}$. 
 \end{enumerate}
\end{lemma}
\begin{proof}
(1) follows from the consistency of $\Delta$ and the
fact that 
{$\funchar{\val}{\alpha}$ and $\funchar{\valtwo}{\alpha}$ each contain a subformula that contradicts the other. }
\Eg\  
if $\val=\vF$ and $\valtwo=\vf$, it cannot be the case that both $\Delta\vdash^{\Lan} \Box\neg\alpha$
and $\Delta\vdash^{\Lan}\Diamond \alpha$ hold.\\ {
For (2), first note that $(\Box \neg \alpha\wedge\Box \alpha)\to\Box\bot$ is an instance of axiom $\mk$ and that $(\Diamond \alpha \to \Box \beta) \to \Box(\alpha \to \beta)$ is a $\mK$-tautology for any $\alpha,\beta$.\\
Suppose that $\val\in\{\vtt, \vff\}$ and $\Delta \vdash^{\Lan} \funchar{\val}{\alpha}$. Thus
$\Delta\vdash^{\Lan}\Box \neg \alpha\wedge\Box \alpha$
and hence $\Delta\vdash^{\Lan}\Box\bot$. 
Assume now that $\Delta \vdash^{\Lan} \funchar{\valtwo}{\beta}$ with $\valtwo \not\in\{\vtt, \vff\}$. Hence, 
 it must be the case that $\Delta\vdash^{\Lan}\beta'$ where $\beta' \in\{\Diamond \beta, \Diamond \neg \beta\}$.
Suppose $\beta'=\Diamond \beta \equiv \neg\Box\neg \beta$. 
Since $\Delta\vdash^{\Lan}(\Diamond \beta \to \Box \bot) \to \Box(\beta \to \bot)$, by $\mpo$ we have $\Delta \vdash^{\Lan}\Box \neg \beta$, which is a contradiction. 
The case where $\beta'=\Diamond \neg \beta$ is similar.}
\qed
\end{proof}

 \subsection{Matrices and Level Semantics}\label{sec:matrices}
{Before introducing the Nmatrices for all the logics in the modal cube, it is important to note that not all eight truth values are required for every logic. The set of values used depends on the specific axioms that characterize each system.
For example, consider the axiom $\mT = \Box\alpha \to \alpha$. This axiom rules out certain combinations of truth values: specifically, any valuation that assigns to formulas 
the values $\vff$, $\vfff$, $\vtt$ or $\vttt$ would violate the axiom by allowing both $\Box \alpha$ and $\neg \alpha$ to hold simultaneously. Thus, such valuations render the axiom unsound. As a result, these four values cannot appear in any model of the logic $\mKT$ or its extensions, and the semantics effectively collapses to the four-valued framework proposed by Kearns in~\cite{DBLP:journals/jsyml/Kearns81}.}
Similarly, the axiom $\mD = \Box \alpha \to \Diamond \alpha$ excludes the values $\vff$ and $\vtt$ from the logic $\mKD$, since $\Box \alpha$ and $\Box \neg\alpha$ cannot both hold. This restriction gives rise to the six-valued semantics presented in~\cite{DBLP:journals/jancl/ConiglioCP15}.

\begin{table}[!t]
    {\scriptsize
    \begin{subtable}[t]{\textwidth}
        \qquad
    \begin{minipage}{.1\linewidth}
        \centering
        \begin{tabular}{|c|c|c|}
            \hline
            $\tilde \bot$ & $\{\vF, \vff \}$ \\ \hline
            \end{tabular}
	\end{minipage}
    \qquad
    \begin{minipage}{.40\linewidth}
      \centering
      \begin{tabular}{|c|c|c|c|c|c|c|c|c|}
\hline
        $\alpha \tilde\to \beta$ & $\vF$   & $\vf$        & $\vff$  & $\vfff$ & $\vttt$ & $\vtt$ & $\vt$      & $\vT$ \\ \hline
        $\vF$              & $\{\vT\}$   & $\{\vT\}$        & $\{\vT\}$   & $\{\vT\}$   & $\{\vT\}$   & $\{\vT\}$  & $\{\vT\}$      & $\{\vT\}$ \\ \hline
        $\vf$              & $\{\vt\}$   & $\{\vT, \vt\}$   & $\{\vtt\}$  & $\{\vT\}$   & $\{\vt\}$   & $\{\vT\}$  & $\{\vT, \vt\}$ & $\{\vT\}$ \\ \hline
        $\vff$             & $\{\vttt\}$ & $\{\vt\}$        & $\{\vtt\}$  & $\{\vT\}$   & $\{\vttt\}$ & $\{\vtt\}$ & $\{\vt\}$      & $\{\vT\}$ \\ \hline
        $\vfff$            & $\{\vttt\}$ & $\{\vt\}$   & $\{\vtt\}$  & $\{\vT\}$   & $\{\vttt\}$ & $\{\vtt\}$ & $\{\vt\}$      & $\{\vT\}$ \\ \hline
        $\vttt$            & $\{\vfff\}$ & $\{\vfff\}$      & $\{\vfff\}$ & $\{\vfff\}$ & $\{\vT\}$   & $\{\vT\}$  & $\{\vT\}$      & $\{\vT\}$ \\ \hline
        $\vtt$             & $\{\vF\}$   & $\{\vf\}$        & $\{\vff\}$  & $\{\vfff\}$ & $\{\vttt\}$ & $\{\vtt\}$ & $\{\vttt\}$    & $\{\vT\}$ \\ \hline
        $\vt$              & $\{\vf\}$   & $\{\vf, \vfff\}$ & $\{\vfff\}$ & $\{\vfff\}$ & $\{\vt\}$   & $\{\vT\}$  & $\{\vT, \vt\}$ & $\{\vT\}$ \\ \hline
        $\vT$              & $\{\vF\}$   & $\{\vf\}$        & $\{\vff\}$  & $\{\vfff\}$ & $\{\vttt\}$ & $\{\vtt\}$ & $\{\vt\}$      & $\{\vT\}$ \\ \hline
        \end{tabular}
    \end{minipage}\caption{$\tilde\bot$ and $\tilde\to$ for all the families, {where $\alpha$/$\beta$-values are the rows/columns.} \label{table:imp}}
    \end{subtable}

    \begin{subtable}[b]{.60\textwidth}
        \centering
\begin{tabular}{|c|c|c|c|c|c|}
\hline
$\alpha$ & $\Box^{\mK} \alpha$     & $\Box^{\mKB} \alpha$    & $\Box^{\mKfour} \alpha$  & $\Box^{\mKfive} \alpha$ & $\Box^{\mKfourfive} \alpha$ \\ \hline
$\vF$    & $\{ \vF, \vf, \vfff \}$ & $\{ \vF \}$             & $\{ \vF, \vf, \vfff \}$ & $\{ \vF \}$             & $\{ \vF \}$                \\ \hline
$\vf$    & $\{ \vF, \vf, \vfff \}$ & $\{ \vF \}$             & $\{ \vF, \vf, \vfff \}$ & $\{ \vF \}$             & $\{ \vF \}$                \\ \hline
$\vff$   & $\{ \vtt \}$            & $\{ \vtt \}$            & $\{ \vtt \}$            & $\{ \vtt \}$            & $\{ \vtt \}$               \\ \hline
$\vfff$  & $\{ \vT, \vt, \vttt \}$ & $\{ \vttt \}$           & $\{ \vT \}$             & $\{ \vT, \vttt \}$      & $\{ \vT \}$                \\ \hline
$\vttt$  & $\{ \vF, \vf, \vfff \}$ & $\{ \vF, \vf, \vfff \}$ & $\{ \vF, \vf, \vfff \}$ & $\{ \vF \}$             & $\{ \vF \}$                \\ \hline
$\vtt$   & $\{ \vtt \}$            & $\{ \vtt \}$            & $\{ \vtt \}$            & $\{ \vtt \}$            & $\{ \vtt \}$               \\ \hline
$\vt$    & $\{ \vF, \vf, \vfff \}$ & $\{ \vF, \vf, \vfff \}$ & $\{ \vF, \vf, \vfff \}$ & $\{ \vF \}$             & $\{ \vF \}$                \\ \hline
$\vT$    & $\{ \vT, \vt, \vttt \}$ & $\{ \vT, \vt, \vttt \}$ & $\{ \vT \}$             & $\{ \vT, \vttt \}$      & $\{ \vT \}$                \\ \hline
\end{tabular}
        \caption{$\tilde\Box$ for family $\mK\star$\label{table:multiKBox}}
    \end{subtable}
    \qquad 
    \begin{subtable}[b]{.4\textwidth}
        \centering
    \begin{tabular}{|c|c|c|}
        \hline
        $\alpha$   & $\Box^{\mKBfourfive} \alpha$  \\ \hline
        $\vF$   & $\{ \vF \}$            \\ \hline
        $\vf$   & $\{ \vF \}$            \\ \hline
        $\vff$ & $\{ \vtt \}$           \\ \hline
        $\vtt$ & $\{ \vtt \}$           \\ \hline
        $\vt$   & $\{ \vF \}$            \\ \hline
        $\vT$   & $\{ \vT \}$            \\ \hline
        \end{tabular}
        \caption{$\tilde\Box$ for family $\mKBfourfive$\label{table:multiKBforfive}}
    \end{subtable}
    \begin{subtable}[b]{.59\textwidth}
        \centering
\begin{tabular}{|c|c|c|c|c|c|}
\hline
$\alpha$ & $\Box^{\mKD} \alpha$    & $\Box^{\mKDB} \alpha$   & $\Box^{\mKDfour} \alpha$ & $\Box^{\mKDfive} \alpha$ & $\Box^{\mKDfourfive} \alpha$ \\ \hline
$\vF$    & $\{ \vF, \vf, \vfff \}$ & $\{ \vF \}$             & $\{ \vF \}$             & $\{ \vF \}$              & $\{ \vF \}$                 \\ \hline
$\vf$    & $\{ \vF, \vf, \vfff \}$ & $\{ \vF \}$             & $\{ \vF, \vf, \vfff \}$ & $\{ \vF \}$              & $\{ \vF \}$                 \\ \hline
$\vfff$  & $\{ \vT, \vt, \vttt \}$ & $\{ \vttt \}$           & $\{ \vT \}$             & $\{ \vT, \vttt \}$       & $\{ \vT \}$                 \\ \hline
$\vttt$  & $\{ \vF, \vf, \vfff \}$ & $\{ \vF, \vf, \vfff \}$ & $\{ \vF \}$             & $\{ \vF \}$              & $\{ \vF \}$                 \\ \hline
$\vt$    & $\{ \vF, \vf, \vfff \}$ & $\{ \vF, \vf, \vfff \}$ & $\{ \vF, \vf, \vfff \}$ & $\{ \vF \}$              & $\{ \vF \}$                 \\ \hline
$\vT$    & $\{ \vT, \vt, \vttt \}$ & $\{ \vT, \vt, \vttt \}$ & $\{ \vT \}$             & $\{ \vT, \vttt \}$       & $\{ \vT \}$                 \\ \hline
\end{tabular}
        \caption{$\tilde\Box$ for family $\mKD\star$\label{table:multiKDBox}}
    \end{subtable}
    \begin{subtable}[b]{.40\textwidth}
        \centering
\begin{tabular}{|c|c|c|c|c|}
\hline
$\alpha$ & $\Box^{\mKT} \alpha$ & $\Box^{\mKTB} \alpha$ & $\Box^{\mKTfour} \alpha$ & $\Box^{\mKTBfourfive} \alpha$ \\ \hline
$\vF$    & $\{ \vF \}$          & $\{ \vF \}$           & $\{ \vF \}$             & $\{ \vF \}$                  \\ \hline
$\vf$    & $\{ \vF, \vf \}$     & $\{ \vF \}$           & $\{ \vF, \vf \}$        & $\{ \vF \}$                  \\ \hline
$\vt$    & $\{ \vF, \vf \}$     & $\{ \vF, \vf \}$      & $\{ \vF, \vf \}$        & $\{ \vF \}$                  \\ \hline
$\vT$    & $\{ \vT, \vt \}$     & $\{ \vT, \vt \}$      & $\{ \vT \}$             & $\{ \vT \}$                  \\ \hline
\end{tabular}
        \caption{$\tilde\Box$ for family $\mKT\star$\label{table:multiKTBox}}
    \end{subtable}
}
\vspace{-0.3cm}
    \caption{Multifunctions for all the logics in the modal cube. \label{table:box}}
\vspace{-0.6cm}
\end{table}

Accordingly, each family of logics introduced in \Cref{rem:logics} corresponds to a distinct subset of truth values, as detailed next.

\begin{definition}[Nmatrices for the modal cube]\label{def:mat-cube}
The set of admissible truth-values for each family of logics is as follows:
\[
\begin{array}{lll}
\mathcal{V}(\mK\star) = \{ \vF, \vf, \vff, \vfff, \vttt, \vtt, \vt, \vT \} & \qquad &
\mathcal{V}(\mKBfourfive) = \{ \vF, \vf, \vff, \vtt, \vt, \vT \} \\
\mathcal{V}(\mKD\star) = \{ \vF, \vf, \vfff, \vttt, \vt, \vT \} & \qquad &
\mathcal{V}(\mKT\star) = \{ \vF, \vf, \vt, \vT \}
\end{array}
\]
The {\em set of values} for a logic $\Lan$
in the family $\rangfamily$, denoted $\V(\Lan)$, 
is $\V(\rangfamily)$;
the set of {\em designated values} of $\Lan$, 
denoted $\mathcal{D}(\Lan)$, is $\mathcal{D}(\Lan) = \mathcal{V}(\rangfamily) \cap \mathcal{D}$; 
and the 
 corresponding {\em set of non-designated values}, denoted $\mathcal{D}^\complement(\Lan)$, is 
  $\mathcal{D}^\complement(\Lan) = \mathcal{V}(\rangfamily) \cap \mathcal{D}^\complement$.
The {\em Nmatrix} $\mathcal{M}$ associated 
to  $\Lan$ 
is determined by its set of values
$\mathcal{V}(\Lan)$, its designated values $\mathcal{D}(\Lan)$,
the non-deterministic functions for $\to$ (restricted to the domain $\V(\Lan)$) and $\bot$ in \Cref{table:imp}, and the
appropriate non-deterministic function for $\Box$  in \Cref{table:box}.
\end{definition}

Matrices for the derived connectives can be defined in the usual way. For
example, the matrix for $\Diamond$ is fully dual to that of $\Box$, obtained
via negation. Note that the matrix for $\neg \alpha$ is already
included in \Cref{table:imp}, as it corresponds to the special case $\alpha \to
\bot$. Finally, we note that our Nmatrices  are refinements of those presented in~\cite{DBLP:journals/flap/SkurtO16}, in the sense that certain valuations--those that would be excluded by level-valuations--have already been omitted.

While Kearns addressed the issue of necessitation in Nmatrices through level-valuations, our approach introduces an additional layer of complexity: it requires a uniform and general definition that applies coherently across all the logics in the modal cube. 
Hence, at level-0, besides considering
    only valuations that follow the corresponding Nmatrix, 
    we also incorporate the condition imposed by the presence of {\em stable values} (see~\Cref{lemma:cons:vals}). 
    Moreover, a  valuation is not removed iff it assigns a designated $\mathcal{N}$-value to tautologies in the previous level.
\begin{definition}[Level-valuation in $\mathcal{M}$]\label{def:level-val}
The {\em $n$-level-valuation} for an Nmatrix $\mathcal{M}$ is a set of total valuations defined inductively as follows:
    \begin{itemize}
    \item $\lv{0}(\mathcal{M}) ~~= \{ v \in \Val{\mathcal{M}}\mid (\exists \alpha,
	v(\alpha)\in\{\vff,\vtt\}) \Rightarrow \forall \beta, v(\beta)\in\{\vff,\vtt\}\}$;
    \item 
  $\lv{n+1}(\mathcal{M}) = \{ v \in \lv{n}(\mathcal{M})\mid  \forall \alpha,~\lsc{n} \alpha \Rightarrow v(\alpha) \in \{ \vT, \vtt \}\}$
\end{itemize}
where $\lsc{n} \alpha$  denotes $\forall w \in \lv{n}(\M), w(\alpha)\in \mathcal{D}$. 
The set of {\em level-valuations} in $\mathcal{M}$ is given by $\lv{}(\mathcal{M}) = \bigcap^{\infty}_{n=0} \lv{n}(\mathcal{M})$. The {\em level-semantic consequence relation} $\models^{\M}$ is defined as: $\Gamma \models^{\M}\alpha \mbox{ iff } \forall v\in\lv{}(\M): (\forall\beta\in\Gamma:v(\beta)\in\D) \Rightarrow v(\alpha)\in\D
$.
A formula $\alpha$ is {\em valid} in $\M$  if $\models^{\M} \alpha$.
\end{definition}

As already noted by Kearns,
$\lv{n}(\M)$ is not  closed under necessitation as illustrated below  for the family $\mKT_\star$ (left)
and family $\mKD_\star$ (right). 

\vspace{0.2cm}
\noindent
\resizebox{.8\textwidth}{!}{
\begin{tabular}{ccc}
\begin{tabular}{|c|c|c|c|c|c|}
\hline
& $p$ \; &$p\to p$ & $\Box(p \to p)$ & $\Box\Box(p \to p)$ & \ldots\\ \hline
1 & $\vF$& $\vT$  &  $\vT$ &  $\{\vT,\vt\}$ & \\ \hline
2 & $\vF$  & $\vT$  &  $\vt$ & $\{\vF,\vf\}$   &\\ \hline
3 & $\vf$   & $\vt$   &  $\{\vF,\vf\}$  & $\{\vF,\vf\}$ &\\ $\vdots$    &    &    & & & \\ \hline
\end{tabular}
&\qquad\qquad&
\begin{tabular}{|c|c|c|c|c|c|}
\hline
& $p$ \; &$p\to p$ & $\Box(p \to p)$ & $\Box\Box(p \to p)$ & \ldots\\ \hline
1 & $\vF$& $\vT$  &  $\vT$ & $\{\vT,\vt,\vttt\}$  & \\ \hline
2 & $\vF$  & $\vT$  &  $\vt$ & $\{\vF,\vf,\vfff\}$ &\\ \hline
3 & $\vF$  & $\vT$  &  $\vttt$ & $\{\vF,\vf,\vfff\}$ &\\ $\vdots$    &    &    & & & \\ \hline
\end{tabular}
\end{tabular}
}
\vspace{0.2cm}

\noindent
Each line in these tables actually represents multiple valuations. For instance, line 1 on the left table corresponds to two valuations: one where
$\Box\Box(p \to p)$ takes the value $\vT$ and another where it takes the value $\vt$. As we move through the levels, some of these valuations are discarded. For example, in the left table, line 3 is eliminated at level 1 because $p \to p$ is a tautology at level 0, and so $\Box(p \to p)$ must hold at level 1. Similarly, line 2 is discarded at level 2. The right table illustrates that, in the family $\mKD_\star$, the $\Box$ non-deterministic function can map designated values to non-designated ones.

Unlike earlier definitions of level-valuations, our definition of $\lv{0}$ 
includes an extra constraint on the values $\vtt$ and $\vff$.
Such a constraint is a natural consequence of the statement (2) in 
\Cref{lemma:cons:vals}. Furthermore, in logics that include the axioms $\mt$ or $\md$, the values $\vtt$ and $\vff$ are excluded altogether, and the definition reduces to the standard one in the literature. 
Finally, the definition of $\lv{n}$ for $n \geq 1$ refines and unifies previous
approaches as it considers the only truth values  ``persist''
under necessitation: for every non-deterministic function $\tilde\Box$ in \Cref{def:mat-cube}, 
if $\val \in \{\vT,\vtt\}$ then $\tilde\Box(\val)\in \D$ (which is not the case for $\vt$ and
$\vttt$).

We are now ready to state and prove soundness. 
\begin{theorem}[Soundness]\label{thm:soundness}
For  $\Lan$, 
$\Gamma \vdash^{\Lan} \alpha$ implies $\Gamma\lsc{} \alpha$.
\end{theorem}
\begin{proof}
    We first show (\textbf{Part I}), by induction on $n$, that if $\vdash^{\Lan}\alpha$ with length $n+1$ then $\lsc{n}\alpha$. 
 In the base case ($n=0$), $\alpha$ is an axiom and it is straightforward to verify that all axioms are tautologies, \ie, $\lsc{0}\alpha$.
 In the inductive step, 
assume the statement holds for all proofs of length less than $n+1, n\geq 0$. 
\begin{itemize}
\item Suppose $\alpha$ is obtained by $\mpo$ to $\gamma \to \alpha$ and $\gamma$. By induction, $\lsc{n}\gamma \to \alpha$ and $\lsc{n}\gamma$. Hence $v(\gamma),v(\gamma\to\alpha)\in\{\vT,\vtt\}$ for any $v \in \lv{n+1}$. 
A direct inspection of the truth table for implication (see \Cref{table:imp}) shows that in this case, $v(\alpha) \in  \D$.

\item Suppose $\alpha = \Box\gamma$ is derived by $\nec$. Again, by the inductive hypothesis, $\lsc{n}\gamma$ and then $v(\gamma)\in\{\vT,\vtt\}$. But this implies that $v(\Box\gamma) \in  \D$.
\end{itemize}
We now generalize the result above. 
Assume that  $\Gamma\vdash^{\Lan}\alpha$ with length $n+1$. We will prove that $\Gamma\lsc{n}\alpha$. There are two cases to consider.
\begin{itemize}
    \item $\alpha$ is a tautology. Hence $\vdash^{\Lan}\alpha$ and this case follows directly from (\textbf{Part I}). 
    \item There are $\gamma_1,\ldots,\gamma_k\in\Gamma$ such that $\vdash^{\Lan}\bigwedge_{i=1}^k \gamma_i\to\alpha$. From (\textbf{Part I}), $\lsc{n}\bigwedge_{i=1}^k \gamma_i\to\alpha$ and hence $\Gamma\lsc{n}\alpha$.
\end{itemize}
Since every valuation is, in particular, an $n$-level valuation for some $n$, it follows that if $\Gamma\vdash^{\Lan}\alpha$ then 
$\Gamma\lsc{}\alpha$.
 \qed
\end{proof}
 \subsection{Completeness of the Level Semantics}\label{sec:comp}

In this section, we prove the completeness of the level semantics.
The proof of this theorem relies on the
well-known Lindenbaum-\L o\'s construction method, showing that every consistent set can be
extended to a maximally consistent set. In the level semantics, a key step in
this construction is to define a \emph{characteristic function}, which we later
show to be a valid level valuation. Interestingly enough, 
we can define such a function directly 
from the interpretation of the semantic values presented in
\Cref{table:truthVals} (see the notation $\val(\alpha)$ in \Cref{def:func}). 

\begin{definition}[Characteristic function for Nmatrices]\label{def:char-func}
    Let $\Lan$ be a modal logic and $\Delta$ be a maximally consistent set in 
    $\Lan$. The {\em characteristic function} $\valchar : \For \to \mathcal{V}(\Lan)$
    is defined as 
    $
        \valchar(\alpha) = \val \mbox{ iff } \Delta \vdash^{\Lan} \val(\alpha)
    $. 
\end{definition}

We observe that, by part (1) of \Cref{lemma:cons:vals}, 
$\valchar(\cdot)$ is indeed a well-defined function. Moreover, the characteristic function 
is unique for all logics in a family. For example, the characteristic function for any logic $\Lan$ in the
family  $\mKT\star$ is:\\

$
    \small
\valchar(\alpha) = 
\begin{cases}
\text{$\vF$} & \text{iff $\Delta \vdash^{\Lan} \Box \neg \alpha$} \\
\text{$\vf$} & \text{iff $\Delta \vdash^{\Lan} \neg \alpha$ and $\Delta \vdash^{\Lan} \Diamond \alpha$} \\
\text{$\vt$} & \text{iff $\Delta \vdash^{\Lan} \alpha$ and $\Delta \vdash^{\Lan} \Diamond \neg \alpha$} \\
\text{$\vT$} & \text{iff $\Delta \vdash^{\Lan} \Box \alpha$} \\
\end{cases}
$.\\

\noindent Note that we use a simplified version of the definition. Rather than writing ``$\vT$ iff $\Delta \vdash^{\Lan} \Box \alpha \wedge \alpha \wedge \Diamond \alpha$'', we simply write ``$\vT$ iff $\Delta \vdash^{\Lan} \Box \alpha$''. This simplification is justified since, in this family of systems, the axiom $\mT$ guarantees that from $\Delta \vdash^{\Lan} \Box \alpha$ we can derive both $\Delta \vdash^{\Lan} \alpha$ and $\Delta \vdash^{\Lan} \Diamond \alpha$.

The following
results show that, for any modal logic $\Lan$, the corresponding characteristic
function is a level-valuation. We start by considering
level 0.

\begin{lemma}[Adequacy of $\valchar$]\label{lemma:adq}
For any modal logic $\Lan$
with Nmatrix $\M$, 
and 
     maximally consistent set $\Delta$ in $\Lan$, 
$\valchar \in \Val{\M}$. 
\end{lemma}
\begin{proof}
Let $\alpha=\conn(\alpha_0,\ldots \alpha_n)$ be a formula with main connective
$\conn$. We proceed by case analysis on $\conn$ to show that $v^{\Lan}_{\Delta}(\alpha) \in
\tilde\conn(v^{\Lan}_{\Delta}(\alpha_0), \ldots, v^{\Lan}_{\Delta}(\alpha_n))$.
Consider the case $\conn=\Box$ and assume that $v^{\Lan}_{\Delta}(\alpha_0)$ is
some value $\val$. By definition, 
\begin{equation}
\Delta \vdash^{\Lan} \funchar{\val}{\alpha_0}\tag{1}
\end{equation}
Assume that $\tilde\Box(\val)$ is a set $V$. Our goal is to show that:
$v^{\Lan}_{\Delta}(\Box\alpha_0)\in V$ which, by definition, is equivalent to proving:
\begin{equation}
\Delta \vdash^{\Lan}\bigvee_{\valtwo\in V}\funchar{\valtwo}{\Box \alpha_0}\tag{2}
\end{equation} 
Using
any deductive system for $\Lan$, it suffices to show that (1) implies (2). All
resulting proof obligations are detailed in the Appendix. 
\qed \end{proof}

\begin{lemma}\label{lemma:lval}
For every level $n$ and maximally consistent set $\Delta$, $v^\Lan_{\Delta} \in \lv{n}(\mathcal{M})$.
\end{lemma}
\begin{proof}
    We proceed by induction on $n$. The base case was proved in \Cref{lemma:adq}. 
	Assume that $\lsc{n} \alpha$ for some formula $\alpha$ and suppose that $\alpha$ is not a theorem in $\Lan$, that is,
$\nvdash^{\Lan} \alpha$. Thus there is 
a maximally consistent set $\Theta$ such that $\Theta \nvdash^{\Lan} \alpha$, which means that $\Theta \vdash^{\Lan} \neg \alpha$. Hence, $v_{\Theta}^{\mathcal{L}}(\alpha) \not\in
\mathcal{D}$. However, by IH, $v_{\Theta}^{\mathcal{L}} \in
\mathcal{\Lan}_n(\mathcal{M})$ and, by hypothesis, $\lsc{n}
\alpha$, implying $v_{\Theta}^{\mathcal{L}} (\alpha) \in \mathcal{D}$, a contradiction.

\noindent
Therefore, $\vdash^{\Lan} \alpha$ and, by ($\nec$), $\vdash^{\Lan} \Box \alpha$. Hence, for
every maximally consistent set $\Delta$, $\Delta \vdash^{\Lan} \alpha$ and $\Delta
\vdash^{\Lan} \Box \alpha$. By definition, $v^\Lan_{\Delta} (\alpha) \in
\mathcal{D}$ and $v^\Lan_{\Delta} (\Box \alpha) \in \mathcal{D}$ and hence, 
$v^\Lan_{\Delta} (\alpha) \in \sN \cap \mathcal{D}$ (see the set $\sN$ in \Cref{tab:values}) 
        and $v^\Lan_{\Delta} \in \mathcal{\Lan}_{n+1}$ as needed.\qed
\end{proof}
\begin{theorem}[Completeness]\label{theo:LvalComplete} For every  modal logic $\Lan$ and associated Nmatrix $\M$, if
$ \Gamma \lsc{} \alpha$ then $\Gamma \vdash^{\Lan} \alpha $.
\end{theorem}
\begin{proof}
    If $\Gamma \nvdash^{\Lan} \alpha$ then there is some $\alpha$-saturated set $\Delta$ such that 
    $\Gamma \subseteq \Delta$. By \Cref{lemma:lval}, $v^\Lan_{\Delta}$ is a level valuation and 
 $\Delta \nvdash^{\Lan} \alpha$; hence $v^\Lan_{\Delta} (\alpha) \not\in \mathcal{D}(\Lan)$ and, 
given that $\Delta \vdash^{\Lan} \beta$, for every $\beta \in \Delta$, then $v^\Lan_{\Delta} (\beta) \in \mathcal{D}(\Lan)$. Therefore, $\Gamma \nvDash^{\mathcal{M}} \alpha$.\qed
\end{proof}

\section{Decision Procedures for the Modal Cube}

\label{sec:truth}

The previous section introduced level semantics for all the logics in the modal cube, providing new results for $\mKfour$, $\mKfive$, $\mKfourfive$, $\mKDfour$, $\mKDfive$, and $\mKBfive$. However, in practice, checking level semantics is challenging, as each step involves verifying whether a given formula is a tautology (see \Cref{def:level-val}). To address this, the current section introduces a {\em partial valuation} semantics (see \Cref{def:val}), which offers a sound and complete finitary method for validating formulas across all the logics in the modal cube. 

We
show that such partial valuations can be \emph{extended} to total valuations,
and that these total valuations are also valid level-valuations. This property, known as \emph{analyticity}, is a key step for establishing the soundness of
this semantics. We then prove \emph{co-analyticity}, showing that 
restricting a level-valuation to a closed set of formulas yields a valid partial valuation. 
Using co-analyticity and \Cref{theo:LvalComplete}, we prove completeness of the partial valuation semantics.

Our decision procedures not only generalize those proposed by Gr\"{a}tz in~\cite{DBLP:journals/logcom/Gratz22} by covering the entire modal cube, but also show how such procedures can be {\em systematically} and {\em modularly} derived from the functions that define truth values (\Cref{def:func}). Additionally, we bridge the Kearns interpretation and traditional Kripke semantics by identifying, based on the interpretation of truth values, the dependencies and relational conditions (in the Kripke sense) that partial valuations must satisfy. These conditions determine how partial valuations are extended and provide a modular approach to proving the results in this section. 

Our method thus brings valuations and possible-world semantics into closer alignment by demonstrating that the derived relational conditions correspond to the standard frame conditions of the respective modal logics.

\subsection{From Truth Values to Dependencies and Relational Conditions}\label{sec:rel}

\paragraph{Relational models on Valuations.} 
The algorithm introduced
in~\cite{DBLP:journals/logcom/Gratz22} for building (finite) truth tables
(\ie, partial valuations) consists of two steps. First, compute the full truth
table for a set of formulas closed under subformulas following the Nmatrix.
Since the resulting table is not necessarily  sound (it may assign non-designated
values to tautologies), the algorithm  systematically removes ``non valid''
partial valuations (\ie, rows in the table). 
A key contribution of our work is to 
propose uniform criteria for all the considered logic by identifying  ``non valid'' partial valuations as  those whose
\emph{dependencies} are not properly satisfied, as explained below. 

Consider a partial valuation $\vv$ and a formula 
$\alpha$ such that $\vv(\alpha)=\vt$. This truth value 
indicates
that $\alpha$ is 
\emph{contingently true}: 
$\alpha$ is true but \emph{possibly false} (recall that $\vt(\alpha)=\Diamond \neg \alpha\wedge \alpha\wedge \Diamond \alpha$). 
 Therefore, there must exist 
 a valuation $\vw$ in which $\alpha$ is false (\ie, $\Diamond \neg \alpha$ holds). This means that the column corresponding to $\alpha$ must contain at least one non-designated value. 
If no such a valuation/row exists,then $\alpha$ is not
 contingently true but \emph{necessarily true}---contradicting the assumption that $\vv(\alpha) = \vt$. In that case, $\vv$ is not a valid valuation and must be discarded.
We thus say that assigning $\vv(\alpha)=\vt$
 imposes a \emph{dependency} on $\vv$: it requires the existence of 
 another valuation $\vw$ to satisfy the semantics. The same reasoning
applies to all truth values $\val$ in the sets
 $\sP$ and $\sPN$, due to the presence of the formulas $\Diamond \alpha$ and $\Diamond \neg \alpha$ in its modal
 characterization. 

However, the mere existence of the above valuation $\vw$ is not sufficient: it must properly \emph{support} the values assigned by $\vv$. 
Suppose, for instance, that for some
formula $\beta$, we have $\vv(\beta) \in \sN$, that is, $\beta$ is necessarily true at $\vv$. 
Then, by $\nec$, $\vw(\beta)$ must be a designated value. 
Moreover, the modal axioms of the logic in question may impose stricter requirements. For example, in any extension of $\mKfour$, if $\vv(\beta) \in \sN$, then $\vw(\beta)$ must not only be designated but must itself also belong to $\sN$. 
We refer to these additional constraints as {\em relational conditions} (see Table~\ref{table:boxCond}): they determine the criteria that a valuation $\vw$ must satisfy in order to support a given valuation $\vv$.

Below, we formalize the above intuitions. 

\begin{table}[t]\centering
\centering
\resizebox{0.82\textwidth}{!}{
\begin{tabular}{|ll|l|l|}
\hline
\multicolumn{2}{|c|}{\textbf{Property }}                                                                                       & \multicolumn{1}{c|}{\textbf{Condition}}                            & \multicolumn{1}{c|}{\textbf{Implies}} \\ \hline
\multicolumn{1}{|l|}{\multirow{2}{*}{$\nec$}}          & \multirow{2}{*}{$\emph{any}$}                               & $\vw (\alpha) \in \sN$ and $\vw R \vw'$               & $\vw' (\alpha) \in \mathcal{D}$      \\ \cline{3-4} 
\multicolumn{1}{|l|}{}                                       &                                                             & $\vw (\alpha) \in \sI$ and $\vw R \vw'$               & $\vw' (\alpha) \not\in \mathcal{D}$  \\ \hline
\multicolumn{1}{|l|}{\multirow{2}{*}{$\mt$}} & \multirow{2}{*}{$\Box \alpha \to \alpha$}                   & $\vw (\alpha) \in \sN$                                                 & $\vw (\alpha) \in \mathcal{D}$       \\ \cline{3-4} 
\multicolumn{1}{|l|}{}                                       &                                                             & $\vw (\alpha) \in \sI$                                                 & $\vw (\alpha) \not\in \mathcal{D}$   \\ \hline
\multicolumn{1}{|l|}{\multirow{2}{*}{$\md$}} & \multirow{2}{*}{$\Box \alpha \to \Diamond \alpha$}          & $\vw (\alpha) \in \sN$                                                 & $\vw (\alpha) \in \sP$                 \\ \cline{3-4} 
\multicolumn{1}{|l|}{}                                       &                                                             & $\vw (\alpha) \in \sI$                                                 & $\vw (\alpha) \in \sPN$                \\ \hline
\multicolumn{1}{|l|}{\multirow{2}{*}{$\mb$}} & \multirow{2}{*}{$\alpha \to \Box \Diamond \alpha$}          & $\vw (\alpha) \in \mathcal{D}$ and $\vw R \vw'$     & $\vw' (\alpha) \in \sP$                \\ \cline{3-4} 
\multicolumn{1}{|l|}{}                                       &                                                             & $\vw (\alpha) \not\in \mathcal{D}$ and $\vw R \vw'$ & $\vw' (\alpha) \in \sPN$               \\ \hline
\multicolumn{1}{|l|}{\multirow{2}{*}{$\mfour$}} & \multirow{2}{*}{$\Box \alpha \to \Box \Box \alpha$}         & $\vw (\alpha) \in \sN$ and $\vw R \vw'$               & $\vw' (\alpha) \in \sN$                \\ \cline{3-4} 
\multicolumn{1}{|l|}{}                                       &                                                             & $\vw (\alpha) \in \sI$ and $\vw R \vw'$               & $\vw' (\alpha) \in \sI$                \\ \hline
\multicolumn{1}{|l|}{\multirow{4}{*}{$\mfive$}} & \multirow{2}{*}{$\Diamond \alpha \to \Box \Diamond \alpha$} & $\vw (\alpha) \in \sP$ and $\vw R \vw'$               & $\vw' (\alpha) \in \sP$                \\ \cline{3-4} 
\multicolumn{1}{|l|}{}                                       &                                                             & $\vw (\alpha) \in \sPN$ and $\vw R \vw'$              & $\vw' (\alpha) \in \sPN$               \\ \cline{2-4} 
\multicolumn{1}{|l|}{}                                       & \multirow{2}{*}{$\Box\Box\alpha \to \Box\Box\Box\alpha$}    & $\vw(\alpha), \vw'(\alpha)\in \sN$, $\vw R\vw'$ and ($ \vw R \vw''$ or $ \vw' R \vw''$)& $\vw''(\alpha)\in \sN$                                  \\ \cline{3-4} 
\multicolumn{1}{|l|}{}                                       &                                                             & $\vw(\alpha), \vw'(\alpha)\in \sI$, $\vw R\vw'$ and ($ \vw R \vw''$ or $ \vw' R \vw''$) & $\vw''(\alpha)\in \sI$              \\ \hline
\end{tabular}
}
\vspace{0.2cm}
    \caption{Relational conditions on models. }\label{table:boxCond}
	\vspace{-0.7cm}
    \end{table}

\begin{definition}[Model]\label{def:model}
Let $\M = \langle \V,\D,\Om\rangle$ be an Nmatrix for a logic $\Lan$
and  $\Lambda\subseteq\For$  closed under subformulas.
A {\em pre-model} in $\Lan$ is a pair $\langle \Pi, R\rangle$, where $\Pi\subseteq [\Lambda \to \V]_{\M}$ is a set of partial valuations, 
and $R \subseteq \Pi \times \Pi$ relates valuations s.t for all 
$\vv\in \Pi$ and $\alpha \in \Lambda$:  

\noindent~~(1) If $\vv(\alpha) \in \sP$, then $\exists  \vw \in \Pi$ such that $\vv R\vw$ and $\vw(\alpha) \in \D$;

\noindent~~(2) If $\vv(\alpha) \in \sPN$, then $\exists \vw \in \Pi$ such that $\vv R\vw$ and $\vw(\alpha) \not\in \D$.

	A {\em $\mK$-model} is a pre-model $\langle \Pi, R \rangle$ s.t. $R$ additionally satisfies the condition
$\nec$ in \Cref{table:boxCond}. For a logic $\Lan$, extending $\mK$, an {\em $\Lan$-model} is a $\mK$-model where for all (frame) properties  in \Cref{table:boxCond} that holds in $\Lan$, 
$R$ and all $\vv\in \Pi$ additionally satisfy the  respective relational conditions in \Cref{table:boxCond}.
 We say that $\vv$ is a {\em partial level-valuation} if there exists a $\Lan$-model $\langle \Pi, R \rangle$ s.t. $v\in \Pi$. \qed
\end{definition}
	
The above definition provides a \emph{procedure} to check whether a set
of partial valuations is legal, \ie, it forms a \emph{model}. 
It suffices to check that  whenever $\vv(\alpha)=\val$ is a  ($\sP\cup \sPN$)-value,  
$\vv$ must be supported by 
another valuation $\vw$ (\ie, $\vv R \vw$)
and $\vv$ and $\vw$ must  satisfy the \emph{relational conditions}
imposed by the logic at hand. 
For concreteness, we show  the resulting model conditions in 
 \Cref{fig:models} where the following notation is used.
 \begin{notation}\label{not:arrows}
     Let $\val\in \V, V\subseteq\V$. We will adopt the following notation:
\begin{itemize}
\item $\val \relcond V$: ``$\forall\alpha\in\For$, if $\vv(\alpha)=\val$ and $\vv R\vw$, then $\vw(\alpha)\in V$''.
\item $\val\arrow V$: ``$\forall\alpha\in\For$, if $\vv(\alpha)= \val$, there must exist  $\vw$ s.t. $\vv R\vw$ and  $\vw(\alpha) \in V$''. 
\end{itemize}
Any value $\val$ in $\sP,\sPN$ must be supported in the sense of \Cref{def:model}; that is, $\val \arrow V$ for some $V$. This means that if $\val \relcond V'$ then $V$ in $\val \arrow V$ must satisfy:
 \begin{itemize}
     \item to be a support for a value  $\val \in \sP$, $V=V'\cap \D(\Lan)$;
     \item to be a support for a value  $\val \in \sPN$, $V=V'\cap \D^\complement(\Lan)$.
 \end{itemize}
 \end{notation}
Note that $\vt$ and $\vf$ belong to both $\sP$ and $\sPN$. Therefore, each must be supported by both a designated and a non-designated value. This implies that, in every logic, we have $\vt \arrow V_1$ and $\vt \arrow V_2$ for some $V_1\not= V_2$, and similarly for $\vf$.

 \paragraph{Seriality.} In logics where $\md$ holds--such as extensions of
 $\mKD$ and $\mKT$--the $\md$-condition requires that $\sN$-values
 (respectively, $\sI$-values) must also be $\sP$-values (respectively,
 $\sPN$-values). This excludes from these logics the values $\vtt$ and $\vff$,
 which, as noted, characterize valuations without successors. For the other logics, 
 note that these values are
 not in $\sP$ or $\sPN$, and therefore cannot appear on the left of $\arrow$.
 Moreover, since both are in $\sN$ and $\sI$, the $\nec$-condition implies that
 if $\vv(\alpha) = \vtt$ and $\vv R \vw$, then $\vw(\alpha) \in \D$ and
 $\vw(\alpha) \notin \D$--a contradiction. This justifies the notation $\vtt
 \relcond \bullet$ and $\vff \relcond \bullet$ in \Cref{fig:models}, as a
 valuation containing these values cannot be related to any other (see
 \Cref{lemma:cons:vals}). 

 \paragraph{Reflexivity.}
 In  logics where $\mt$ holds (and hence $\md$), the $\mt$-condition  further excludes
 the values $\vttt$ and $\vfff$ since, \eg, the latter is in $\sN$ but it is not
 designated. Moreover, checking the supports ($\arrow$)
 becomes simpler as shown in \Cref{fig:models}. For instance, 
 in $\mKT$, we do not have to check whether $\vt \arrow \vT,\vt$ since 
  this condition trivially holds by reflexivity.
 Furthermore, observe that, for every $\iota \relcond V$ in logics extending $\mKT$, 
  $\iota \in V$ (which of course is not the case for values $\vttt, \vfff, \vff, \vtt$ in other logics).

 \paragraph{Symmetry and Transitivity.}
 Regarding the $\mb$-condition, note that instead of $\vT\relcond \vT,\vt,\vtt,\vttt$ as in 
 $\mK$, we have $\vT\relcond \vT,\vt$ in $\mKB$ since $\vtt$ and $\vttt$ are not in $\sP$. 
 In logics where the axiom $\mfour$ holds, the  relation is further restricted
 to $\vT \relcond \vT$ (since $\vt$ is not in $\sN$). 

\paragraph{Euclidianity.} In any logic extending $\mKfive$, $\vT
\relcond \vT, \vt$, but $\vt \not\relcond \vT$. For this reason, besides the
condition for axiom $\mfive$, we add a further condition derived from the formula 
$\Box \Box \alpha \to \Box \Box \Box \alpha$ (a theorem in any system extending
$\mKfive$). Its twofold purpose is to ensure that, for every $\mfive$-model:
(1) if $\vw R\vv$ and $\vw R\vv'$, then if $\vw(\alpha) \in \sN$ (or $\sI$) and $\vv(\alpha) = \vT$ (or $\vF$), then $v'(\alpha) = \vT$ (or
$\vF$); (2) if $\vw R\vv$ and $\vv R\vv''$, then if $\vw(\alpha) \in \sN$ (or $\sI$) and
$\vv(\alpha) = \vT$ (or $\vF$), then $\vv''(\alpha) = \vT$ (or $\vF$). This
guarantees that $\vw R\vv \land \vw R\vv'$ implies $\vv R\vv'$.

A fundamental property of our model is that the relation $\relcond$ induced by any logic $\Lan$ satisfies the standard frame conditions associated with $\Lan$. This marks a significant step toward reconciling Kripke semantics with Kearns' semantics--a completely new feature not present in previous works.

 \begin{theorem}[Frame Properties]\label{th:frame}
     (1) If axiom $\mt/\mb/\mfour/\mfive$ is valid in the logic $\Lan$, 
     then the induced $\relcond$ is
     reflexive/symmetric/transitive/Euclidian, respectively; 
     (2) If axiom $\md$ is valid in $\Lan$,
     for all $\val\in \V(\Lan)$, $\val \arrow V$ for some $V$; and 
     (3) If axiom $\md/\mt/\mb/\mfour/\mfive$ is valid in  $\Lan$, 
     and $\langle \Pi, R\rangle$ is an $\Lan$-model, 
     then there exists a $\Lan$-model $\langle \Pi, R'\rangle$ s.t. $R' \supseteq R$ and 
     $R'$ is serial/reflexive/symmetric/transitive/Euclidian. 
 \end{theorem}

 \begin{proof}
     Statement (1) follows by inspecting the 
     induced $\relcond$ for each $\Lan$.
     For symmetry in  $\mKB$, we can check 
     that  $\val\relcond V_\val$ implies that for all $\valtwo\in V_\val$, $\valtwo\relcond V_\valtwo$ and 
     $\val\in V_\valtwo$. 
	 For euclidianness in (extensions of) $\mKfive$, the interesting case is when 
	  $\kappa \in \sN \cup \sI$,
	and it is solved as explained above due to the condition imposed by $\Box \Box \alpha \to \Box \Box \Box \alpha$.
For (2), note
that in all extensions of $\mKD$ (including logics where $\mt$ holds), 
     all the values of the logic appear on the left of $\arrow$.
For (3), note that $R$ in  $\langle \Pi, R\rangle$ 
	does not 
     necessarily 
     satisfy the usual frame conditions. However, due to (1), $R$ can 
     always  by extended to a 
     $R'\supseteq R$ that satisfies such conditions. \qed
 \end{proof}
\begin{figure}[!t]
    \centering
\resizebox{.96\textwidth}{!}{
    \reltable{$\mK$}{.20}
    {
        \vT &\arrow \vT, \vt, \vtt, \vttt \\
        \vt &\arrow \vT, \vt, \vtt, \vttt \\
        \vt &\arrow \vF, \vf, \vff, \vfff \\
        \vttt &\arrow \vF, \vf, \vff, \vfff \\
        \vfff &\arrow \vT, \vt, \vtt, \vttt \\
        \vf &\arrow \vT, \vt, \vtt, \vttt \\
        \vf &\arrow \vF, \vf, \vff, \vfff \\
        \vF &\arrow \vF, \vf, \vff, \vfff
    }    
    {
        \vT &\relcond \vT, \vt, \vtt, \vttt \\
        \vtt &\relcond \bullet \\
        \vttt &\relcond \vF, \vf, \vff, \vfff \\
        \vfff &\relcond \vT, \vt, \vtt, \vttt \\
        \vff &\relcond \bullet \\
        \vF &\relcond \vF, \vf, \vff, \vfff
    }
    \reltable{$\mKB$}{0.18}
        {
            \vT &\arrow \vT, \vt \\
            \vt &\arrow \vT, \vt \\
            \vt &\arrow \vf, \vfff \\
            \vttt &\arrow \vf, \vfff \\
            \vfff &\arrow \vt, \vttt \\
            \vf &\arrow \vt, \vttt \\
            \vf &\arrow \vF, \vf \\
            \vF &\arrow \vF, \vf
        }    
        {
            \vT &\relcond \vT, \vt \\
            \vt &\relcond \vT, \vt, \vf, \vfff \\
            \vtt &\relcond \bullet \\
            \vttt &\relcond \vf, \vfff \\
            \vfff &\relcond \vt, \vttt \\
            \vff &\relcond \bullet \\
            \vf &\relcond \vF, \vf, \vt, \vttt \\
            \vF &\relcond \vF, \vf
        }
        \reltable{$\mKfour$}{0.19}
        {
            \vT &\arrow \vT, \vtt \\
            \vt &\arrow \vT, \vt, \vtt, \vttt \\
            \vt &\arrow \vF, \vf, \vff, \vfff \\
            \vttt &\arrow \vF, \vff \\
            \vfff &\arrow \vT, \vtt \\
            \vf &\arrow \vT, \vt, \vtt, \vttt \\
            \vf &\arrow \vF, \vf, \vff, \vfff \\
            \vF &\arrow \vF, \vff
        }    
        {
            \vT &\relcond \vT, \vtt\\
            \vtt &\relcond \bullet \\
            \vttt &\relcond \vF, \vff \\
            \vfff &\relcond \vT, \vtt \\
            \vff &\relcond \bullet \\
            \vF &\relcond \vF, \vff
        }
        \reltable{$\mKfive$}{0.14}
        {
            \vT &\arrow \vT, \vt \\
            \vt &\arrow \vt \\
            \vt &\arrow \vf \\
            \vttt &\arrow \vF, \vf \\
            \vfff &\arrow \vT, \vt \\
            \vf &\arrow \vt \\
            \vf &\arrow \vf \\
            \vF &\arrow \vF, \vf
        }    
        {
            \vT &\relcond \vT, \vt \\
            \vt &\relcond \vt, \vf \\
            \vtt &\relcond \bullet \\
            \vttt &\relcond \vF, \vf \\
            \vfff &\relcond \vT, \vt \\
            \vff &\relcond \bullet \\
            \vf &\relcond \vt, \vf \\
            \vF &\relcond \vF, \vf
        }
        \reltable{$\mKfourfive$}{0.14}
            {
                \vT &\arrow \vT \\
                \vt &\arrow \vt \\
                \vt &\arrow \vf \\
                \vttt &\arrow \vF \\
                \vfff &\arrow \vT \\
                \vf &\arrow \vt \\
                \vf &\arrow \vf \\
                \vF &\arrow \vF
            }  
            {
                \vT &\relcond \vT \\
                \vt &\relcond \vt, \vf \\
                \vtt &\relcond \bullet \\
                \vttt &\relcond \vF \\
                \vfff &\relcond \vT \\
                \vff &\relcond \bullet \\
                \vf &\relcond \vt, \vf \\
                \vF &\relcond \vF
            }
        \reltable{$\mKBfourfive$}{0.13}
        {
            \vT &\arrow \vT \\
            \vt &\arrow \vt \\
            \vt &\arrow \vf \\
            \vf &\arrow \vt \\
            \vf &\arrow \vf \\
            \vF &\arrow \vF
        }  
        {
            \vT &\relcond \vT \\
            \vt &\relcond \vt, \vf \\
            \vtt &\relcond \bullet \\
            \vff &\relcond \bullet \\
            \vf &\relcond \vt, \vf \\
            \vF &\relcond \vF
        }
    \reltable{$\mKD$}{0.15}
        {
            \vT &\arrow \vT, \vt, \vttt \\
            \vt &\arrow \vT, \vt, \vttt \\
            \vt &\arrow \vF, \vf, \vfff \\
            \vttt &\arrow \vF, \vf, \vfff \\
            \vfff &\arrow \vT, \vt, \vttt \\
            \vf &\arrow \vT, \vt, \vttt \\
            \vf &\arrow \vF, \vf, \vfff \\
            \vF &\arrow \vF, \vf, \vfff
        }
        {
            \vT &\relcond \vT, \vt, \vttt \\
            \vttt &\relcond \vF, \vf, \vfff \\
            \vfff &\relcond \vT, \vt, \vttt \\
            \vF &\relcond \vF, \vf, \vfff
        }
    }
    \\\vspace{0.2cm}
    \resizebox{.90\textwidth}{!}{
    \reltable{$\mKDB$}{0.16}
        {
            \vT &\arrow \vT, \vt \\
            \vt &\arrow \vT, \vt \\
            \vt &\arrow \vf, \vfff \\
            \vttt &\arrow \vf, \vfff \\
            \vfff &\arrow \vt, \vttt \\
            \vf &\arrow \vt, \vttt \\
            \vf &\arrow \vF, \vf \\
            \vF &\arrow \vF, \vf
        }
        {
            \vT &\relcond \vT, \vt \\
            \vt &\relcond \vT, \vt, \vf, \vfff \\
            \vttt &\relcond \vf, \vfff \\
            \vfff &\relcond \vt, \vttt \\
            \vf &\relcond \vF, \vt, \vf, \vttt \\
            \vF &\relcond \vF, \vf
        }
     \reltable{$\mKDfour$}{0.15}
        {
            \vT &\arrow \vT \\
            \vt &\arrow \vT, \vt, \vttt \\
            \vt &\arrow \vF, \vf, \vfff \\
            \vttt &\arrow \vF \\
            \vfff &\arrow \vT \\
            \vf &\arrow \vT, \vt, \vttt \\
            \vf &\arrow \vF, \vf, \vfff \\
            \vF &\arrow \vF
        }
        {
            \vT &\relcond \vT \\
            \vttt &\relcond \vF \\
            \vfff &\relcond \vT \\
            \vF &\relcond \vF
        }
    \reltable{$\mKDfive$}{0.14}
        {
            \vT &\arrow \vT, \vt \\
            \vt &\arrow \vt \\
            \vt &\arrow \vf \\
            \vttt &\arrow \vF, \vf \\
            \vfff &\arrow \vT, \vt \\
            \vf &\arrow \vt \\
            \vf &\arrow \vf \\
            \vF &\arrow \vF, \vf
        }
        {
            \vT &\relcond \vT, \vt \\
            \vt &\relcond \vt, \vf \\
            \vttt &\relcond \vF, \vf \\
            \vfff &\relcond \vT, \vt \\
            \vf &\relcond \vt, \vf \\
            \vF &\relcond \vF, \vf
        }
    \reltable{$\mKDfourfive$}{0.12}
        {
            \vT &\arrow \vT \\
            \vt &\arrow \vt \\
            \vt &\arrow \vf \\
            \vttt &\arrow \vF \\
            \vfff &\arrow \vT \\
            \vf &\arrow \vt \\
            \vf &\arrow \vf \\
            \vF &\arrow \vF
        }
        {
            \vT &\relcond \vT \\
            \vt &\relcond \vt, \vf \\
            \vttt &\relcond \vF \\
            \vfff &\relcond \vT \\
            \vf &\relcond \vt, \vf \\
            \vF &\relcond \vF
        }

     \reltableall{$\mKT$}{0.15}
        {
            \vt &\arrow \vF, \vf \\
            \vf &\arrow \vT, \vt
        }
        {
            \vT &\relcond \vT, \vt \\
            \vF &\relcond \vF, \vf
        }
   \reltableall{$\mKTB$}{0.15}
            {
                \vt &\arrow \vf \\
                \vf &\arrow \vt
            }
            {
                \vT &\relcond \vT, \vt \\
                \vt &\relcond \vT, \vt, \vf \\
                \vf &\relcond \vF, \vt, \vf \\
                \vF &\relcond \vF, \vf
            }
    \reltableall{$\mKTfour$}{0.13}
        {
            \vt &\arrow \vF, \vf \\
            \vf &\arrow \vT, \vt
        }
        {
            \vT &\relcond \vT \\
            \vF &\relcond \vF
        }
  \reltableall{$\mKTBfourfive$}{0.13}
            {
                \vt &\arrow \vf \\
                \vf &\arrow \vt
            }
            {
                \vT &\relcond \vT \\
                \vt &\relcond \vt, \vf \\
                \vf &\relcond \vt, \vf \\
                \vF &\relcond \vF
            }
    }\caption{Relational Model.
         The absence of an entry $\val \relcond V$ is interpreted as $\val \relcond \V(\Lan)$ (any value supports $\val$). 
     For an $\Lan$ not extending $\mKT$
     where  $\val\relcond V$, 
 we further have: 
     for any $\val\in \sP$ (and $\val \in \sPN$), 
     $\val \arrow V \cap \D(\Lan)$ 
     (and  $\val \arrow V \cap \D^\complement(\Lan)$).
 \label{fig:models}}
\vspace{-0.7cm}
\end{figure}

\vspace{-0.3cm}
\paragraph{Analyticity.} Recall that Gr\"{a}tz's original algorithm constructs
the entire truth table {\em before} eliminating any rows. This naturally raises
the question: is there a way to generate only valid partial valuations from the
start, thereby avoiding the need for post hoc filtering?

In the truth table method, we start by assigning all possible values to the atomic propositions. A new column for a formula $\alpha$ is introduced only once all its subformulas have been added to the table. When working with the partial valuations defined here, one might worry that adding such columns--and checking the associated constraints--could lead to the elimination of some valuations (\ie, rows in the table), potentially resulting in an empty model. This would be problematic, as it would undermine the entire model construction process.

The following lemmas show, however, that all valid partial valuations in a model can be extended. To ensure this, we propose an algorithm that, in non-deterministic cases, selects one admissible value in a way that guarantees that all rows are preserved and correctly extended.  As a by-product, if the algorithm always selects \emph{every} admissible value (instead of only one), it will build a set containing all and only the valid partial valuations, which corresponds to the filtered truth table of Gr\"atz's  original algorithm.

\begin{notation}[Extensions]
Assume that $\Lambda\subseteq \For$ and $\Lambda \cup \{\beta\}$ are both closed
under subformulas, and let $\vv : \Lambda \to \V$ and $\vv':\Lambda\cup\{\beta\}\to \V$ be
two partial valuations. We say that $\vv'$ extends $\vv$, notation $\vv \extsym
\vv'$, if $\vv(\alpha)=\vv'(\alpha)$ for all $\alpha\in \Lambda$.
Let 
$M=\langle \Pi, R\rangle$  and 
$M'=\langle \Pi', R'\rangle$ be $\Lan$-models defined, respectively on 
$\Lambda$ and  $\Lambda\cup\{\beta\}$. 
We say that $M'$ extends $M$, notation $M\extsym M'$, 
if either $\Pi=\emptyset$ or  there is a bijection $f : \Pi \to \Pi'$
such that $\vv \extsym f(\vv)$ and
$\vv R \vw$ iff $f(\vv)R' f(\vw)$. 
\end{notation}

\begin{lemma}[Extension Lemma]\label{lemma:extlemma}
    Let  $\Lambda\subseteq\For$ and $\Lambda \cup \{\alpha\}$
    be closed under subformulas, and 
    $\langle \Pi, R\rangle$ be a $\Lan$-model 
    where $\Pi \subseteq [\Lambda \to \V]_\M$. 
    Then, there exists $\langle \Pi', R'\rangle$
    s.t. $\Pi' \subseteq [\Lambda\cup\{\alpha\} \to \V]_\M$
    and $\langle \Pi, R\rangle \extsym \langle \Pi', R'\rangle$.
\end{lemma}
\begin{proof}
    If $\alpha\in\At$, we consider two cases: 
    If $\Pi=\emptyset$, then 
    the needed model is the  set of 
    valuations with domain $\alpha$ ($\Pi = [\{\alpha\} \to \V]_\M$)
    and $R=\{(\vv_1,\vv_2)\mid \vv_1(\alpha) \arrow \vv_2(\alpha)\}$;
    otherwise, we choose any $\beta\in \Lambda$ and 
    extend all $\vv$ with $\vv(\alpha)=\vv(\beta)$. It is easy to show that
    ``copying'' a column is a valid extension. 

    \noindent Consider the case of $\Box\alpha$ with $\alpha \in \mathrm{dom}(\Pi)$. Since the interpretation $\tilde\Box$ is not necessarily deterministic, we must ensure that all partial valuations in the model are extended in a \emph{deterministic} manner. The extension procedure depends on the logic in question. For example, in the case of $\mKT$, the situation is as follows:
    \vspace{-1.1cm}
        \begin{multicols}{2}
        \valCondTwoOneLine
        {\alpha}{\Box \alpha}
        {\vf, \vt}
        {\vF ? \vf ?}

        \valCondTwoOneLine
        {\alpha}{\Box \alpha}
        {\vT}
        {\vT ? \vt ?}

    \end{multicols}
    \vspace{-0.8cm}
\noindent
meaning that if $\vv(\alpha)\in \{\vf,\vt\}$, the non-deterministic function dictates that the 
    resulting value is either $\vF$ or $\vf$ (similarly for the figure on the right). 
    The criterion to extend $\vv$ is the following: If there are  $\vv', \vv''$ 
    s.t. $\vv R \vv'$, $\vv R\vv''$, $\vv'(\alpha)=\vT$ and $\vv'(\alpha)\neq\vT$, then 
    extend with the lowercase values; otherwise with the uppercase values. The 
    rationale is that, if there is no $v'$, $vRv'$, such that $v'(\alpha) = \vT$ or $v'(\alpha) \neq \vT$ then, in the first case, $\Box \alpha$ is always false (hence, impossible), while, in the second case (which corresponds to choose $\vT$ in the figure on the right), $\Box \alpha$ is always true (hence, necessary). Similarly, if there are both $v'(\alpha) = \vT$ and $v'(\alpha) \neq \vT$, then we can safely infer that $\Box \alpha$ is contingently false or true (lowercase values).
    The other cases assume similar criteria, and they can be found in the Appendix.

    \noindent Consider the case $\alpha \to \beta$ for  $\mKD$ (similarly for all logics since they share $\tilde\to$):
    \vspace{-1.0cm}
    \begin{multicols}{3}
        \valCondThreeOneLine
        {\alpha}{\beta}{\alpha \to \beta}
        {\vf, \vfff}
        {\vf}
        {\vT ? \vt ?}

        \valCondThreeOneLine
        {\alpha}{\beta}{\alpha \to \beta}
        {\vt, \vf}
        {\vt}
        {\vT ? \vt ?}        

    \valCondThreeOneLine
    {\alpha}{\beta}{\alpha \to \beta}
    {\vt}
    {\vf}
    {\vf ? \vfff ?}

    \end{multicols}
    \vspace{-0.8cm}

    The criterion is: if there is $v'$ such that $vRv'$, $v'(\alpha) \in \mathcal{D}$ 
	and $v'(\beta) \not\in \mathcal{D}$, then extend with $\vt$ and $\vf$; 
otherwise, extend with $\vT$ and $\vfff$. Since the implication is classical,
$\alpha \to \beta$ is falsified only when $\alpha$ is designated but $\beta$ is
not. 
 In the first case, the implication  $\alpha \to \beta$ is contingently true ($\vt$)
and $v$ can be extended with $\vt$ and $\vf$. 
In the second case, the implication is not falsified
and we can infer that $\alpha \to \beta$ is
necessarily true ($\vT$). Observe that, although $\vfff$ is a
non-designated value, it is a $\sN$-value, and then,  
 if $\alpha \to \beta$ is always true ``onward'', then $\alpha \to \beta$ is
false ``now'' but necessarily true ($\vfff$).

\end{proof}

\begin{theorem}[Analyticity]\label{th:analy}
    Every partial level-valuation can be extended to a level-valuation.
\end{theorem}

\begin{definition}[Partial Valuation Semantics]
Given $\Gamma \subseteq \For$, we write $\closed{\Gamma}$ to denote the smallest set containing $\Gamma$ that is closed under subformulas.
    We say that $\alpha$ is a {\em partial-semantic consequence} of $\Gamma$ in logic $\Lan$, 
    notation $\Gamma \vDash^{\Lan}\alpha $, iff
    
$\forall \vv\in[\closed{\Gamma\cup\{\alpha\}}\to \V]_\M: (\forall\beta\in\Gamma:\vv(\beta)\in\D) \Rightarrow \vv(\alpha)\in\D$. 
\end{definition}

\begin{theorem}[Soundness]\label{lem:sound-PV} 
    For every  $\Lan$,  $\Gamma \vdash^{\Lan} \alpha \Rightarrow \Gamma \vDash^{\Lan} \alpha$. 
\end{theorem}
\begin{proof}
We proceed by induction on the derivation of $\Gamma \vdash^{\Lan} \alpha$. The case for $\mpo$ relies on \Cref{lemma:extlemma}, which ensures that any valuation $\vv$ with domain $\closed{\Gamma \cup \{\alpha\}}$ can be extended to the  domain $\closed{\Gamma \cup \{\alpha \to \beta\}}$.
\end{proof}

\paragraph{Co-analyticity.}
The completeness proof for the partial valuation semantics relies on \Cref{theo:LvalComplete}. Accordingly, we must show that level (\ie, full) valuations, when restricted to a set $\Lambda$ closed under subformulas, yield valid partial valuations. We denote the restriction of a valuation $\vv$ to the domain $\Lambda$ by $\valrestriction{\vv}{\Lambda}$. 

\begin{theorem}[Co-analyticity]\label{th:co-analy}
    For every level-valuation  $\vv \in \lv{}(\mathcal{M})$
    and every set closed under subformulas $\Lambda$, $\valrestriction{\vv}{\Lambda}$ is a
partial level-valuation. 
\end{theorem}
\begin{proof}
	Following Gr{\"{a}}tz, we consider 
\emph{partial'} valuations~\cite{DBLP:journals/logcom/Gratz22}, where
 partial valuations are supported by 
level-valuation.
 By systematically reducing level to partial valuation,
the two definitions of partial valuations are shown to be equivalent. Now we prove that level
valuations are indeed partial' valuations.
Given a formula $\alpha$, let $\alpha[\val]$ denote the values that result when evaluating 
the corresponding multifunctions in the connectives of $\alpha$
with $\val$  as in  $(\alpha\wedge \Box \alpha)[\vf]=\bigcup_\val\tilde\wedge(\vf, \val)$
where $\val\in\tilde\Box(\vf)$. 
Let $\vv$ be a
level-valuation and assume that
$\vv_p=\valrestriction{\vv}{\Lambda}$ is not partial', \ie, 
 there exists $\alpha\in \Lambda$ s.t. $\vv(\alpha)=\val$ is not supported. 
 Assume that $\iota \arrow
V_\iota$,  $\kappa \relcond V_\kappa$ in  $\Lan$
and the following situation
\begin{center}
    \includegraphics[scale=0.7]{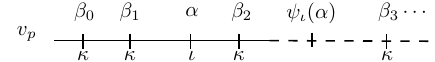}
    \end{center}
To show that there is no $\vw$ s.t. 
    $\vv R\vw$, 
     we choose $\valtwo$, $\Delta=\{\beta_0\cdots\}$,  and $\psi_\val(\alpha)$ s.t.
$\psi_\val(\alpha)[\val] \in \{\vt,\vttt\}$,  $\beta_i[\kappa] = \{\vT\}$, and 
    $\vw(\alpha) \in V_\iota$ iff $w(\psi_\iota(\alpha)) \not\in \mathcal{D}$ (if $\Delta = \emptyset$ the proof below is the same.) 
The choices are guided by  
the modal characterization of the truth values (e.g., 
in $\mKT$ $\psi_\vf(\alpha)=\vF(\alpha)\vee \vf(\alpha)$, the other cases are in the Appendix. 
The idea is to show that, 
if $\vw$ respects the $\relcond$-conditions, then it cannot assign the expected
supporting value in $\vw(\alpha)$,
 i.e., 
     $w(\psi_\iota(\alpha)) \not\in \mathcal{D}$ iff there is some $i \geq 0$ such that $w(\beta_i) \not\in \mathcal{D}$. Hence,  
    $\Delta \models^{\Lan} \psi_\iota(\alpha)$.
    By completeness (of level-semantics) $\Delta \vdash^{\Lan} \psi_\iota(\alpha)$, 
    and by compactness of $\Lan$, 
     $\Delta' \vdash^{\Lan}\psi_\iota(\alpha)$ 
     for a finite    $\Delta \supseteq \Delta'=\{\gamma_0,\cdots,\gamma_n\}$. 
     By $\nec$ and the deduction theorem, 
    $ \vdash^{\Lan}\Box(\gamma_0 \to \cdots \to \gamma_n \to \psi_\iota(\alpha))$.
    By soundness (of level-semantics), $\models^{\Lan}\Box(\gamma_0 \to \cdots \to \gamma_n \to \psi_\iota(\alpha))$. 
We then find a contradiction since, using the Nmatrices, we can check that $\vv(\Box(\gamma_0 \to \cdots \to \gamma_n \to \psi_\iota(\alpha))) \not\in \mathcal{D}$.

\end{proof}

\begin{theorem}[Completeness]
    For every $\Lan$,  $\Gamma \vDash^{\Lan} \alpha \Rightarrow 
     \Gamma \vdash^{\Lan} \alpha$. 
\end{theorem}

\section{Concluding Remarks}\label{sec:conc}

This paper generalizes previous non-deterministic semantic approaches to
the entire modal cube of  normal modal logics. First, 
Kearns' semantics based on Nmatrices with level valuations was
extended to the full modal cube, starting from an 8-valued Nmatrix for the
basic logic
$\mK$. 
Second, the decision procedure for $\mKT$ and
$\mSfour$ using partial valuations over a 3-valued Nmatrix proposed by
Gr\"{a}tz, was extended to all systems of the modal cube, again starting from
the 8-valued Nmatrix for $\mK$. 
A distinctive feature of our framework is that it bridges Nmatrix semantics with Kripke semantics by ensuring that the resulting models satisfy the standard frame conditions associated with each modal system.

Our results are supported by prototypical tools~\cite{tool}. An OCaml
implementation generates truth tables and visualizes models, making
the Kripke correspondence explicit (and helping us to prove, \eg,
that modal axioms are indeed tautologies). We have a very
preliminary and partial formalization of our results in Rocq,
initially focusing on~$\mSfive$. Given the central role of analyticity in our approach, we developed a
Maude specification~\cite{DBLP:journals/jlap/DuranEEMMRT20} to explore
all possible extensions (\Cref{lemma:extlemma}) and automatically
verify their validity.

Future work includes extending our framework to intuitionistic and ecumenical~\cite{DBLP:conf/dali/MarinPPS20} modal logics by
combining it with the approach from~\cite{IPL}. Another line of investigation
involves the relationship between Kripke and Nmatrix countermodels. For instance, the minimal
Kripke countermodel for~$\Diamond \alpha \to \Box \alpha$ in~$\mKT$ involves
two worlds, which is mirrored in the matrix semantics by a two-row model
assigning ${\vt, \vF}$ to~$\alpha$. Understanding this correspondence may yield
insights into the complexity-theoretic aspects of the frameworks. Finally, we
intend to describe the Nmatrices and the level valuation restrictions
analytically using swap structures, as done e.g., 
in~\cite{ConiglioPS}. This representation should enable the generalization of 
the decision procedures proposed here to a broader class of modal
systems.
 
\subsubsection{Acknowledgments}
This study was financed, in part, by the S\~{a}o Paulo
Research Foundation (FAPESP), Brasil. Process Number n. 2023/16021-9 and n.
2021/01025-3. 
The work of Olarte was partially funded by  the NATO Science for Peace and Security Programme through grant number  G6133 (project SymSafe)
and the SGR project PROMUEVA (BPIN 2021000100160) under the supervision of Minciencias Colombia. 
Pimentel is supported by the Leverhulme grant RPG-2024-196 and
has received funding from the European Union's Horizon 2020 research and
innovation programme under the Marie Sk\l odowska-Curie grant agreement Number
101007627. 
Coniglio acknowledges support by an individual research grant from
the National Council for Scientific and Technological Development (CNPq,
Brazil), grant 309830/2023-0. He was also supported by the S\~ao Paulo
Research Foundation (FAPESP, Brazil), thematic project {\em
Rationality, logic and probability -- RatioLog}, grant  2020/16353-3.
We are grateful for the useful suggestions from Carlos Caleiro and Sergio Marcelino, as well as the anonymous referees.

\newpage
\bibliographystyle{splncs04}

\newpage 

\appendix

\section{Detailed Proofs}\label{sec:app-proofs}

In this appendix we give some more details of the proofs sketched in the main
body of the paper. We also prove some auxiliary  results. 

\subsection{Proofs in \Cref{sec:level}}\label{app:sec-level}

\begin{proof}[of \Cref{lemma:adq}]
    As noticed in the proof sketch of this lemma, it suffices to show
    the implication (1) $\to$ (2) where (1) is the hypotheses
    we have due to the fact that $v^{\Lan}_{\Delta}(\alpha_0)$
    took a value in a set $V$ and hence, 
    the goal is to prove (2) $\Delta \vdash^{\Lan}\bigvee_{\valtwo\in V}\funchar{\valtwo}{\Box \alpha_0}$.
	All the cases are easy consequences of $\vdash^{\Lan}$. 
    See \Cref{tab:adq1,tab:adq2,tab:adq3} for the proof obligations due to $\Box$
	and \Cref{tab:cfunc-imp} for the cases due to $\to$. 
    \qed
\end{proof}

\begin{table}[h!]\centering
    \resizebox{\textwidth}{!}{

}
\caption{Showing that $v_{\Delta}^{\mK}(\cdot)$ is a level valuation in the $\mK\star$ family.\label{tab:adq1}}
    \end{table}

\begin{table}[h!]\centering
    \begin{subtable}[t]{.47\textwidth}
        \resizebox{\textwidth}{!}{

        \begin{tabular}{|l|l|l|l|}
            \hline
            \multicolumn{1}{|c|}{$v_{\Delta}^\mKBfourfive (\beta)$} & \multicolumn{1}{c|}{Hyp: $\Delta \vdash^{\mKBfourfive}$} & \multicolumn{1}{c|}{$\Box^{\mKBfourfive} (v_{\Delta}^\mKBfourfive(\beta))$} & \multicolumn{1}{c|}{Goal: $\Delta \vdash^{\mKBfourfive}$}                 \\ \hline
            $\vF$                                                   & $\neg \beta, \Box \neg \beta, \Diamond \neg \beta$       & $\Box^{\mKBfourfive} (\vF) = \{ \vF \}$,                                    & $\neg \Box \beta \land \Box\neg\Box \beta \land \Diamond \neg \Box \beta$ \\ \hline
            $\vf$                                                   & $\neg \beta, \Diamond \neg \beta, \Diamond \beta$        & $\Box^{\mKBfourfive} (\vf) = \{ \vF \}$                                     & $\neg \Box \beta \land \Box\neg\Box \beta \land \Diamond \neg \Box \beta$ \\ \hline
            $\vff$                                                  & $\neg \beta, \Box \beta, \Box \neg \beta$                & $\Box^{\mKBfourfive} (\vff) = \{ \vtt \}$                                   & $\Box \beta \land \Box \Box \beta \land \Box \neg \Box \beta$             \\ \hline
            $\vtt$                                                  & $\beta, \Box \beta, \Box \neg \beta$                     & $\Box^{\mKBfourfive} (\vff) = \{ \vtt \}$                                   & $\Box \beta \land \Box \Box \beta \land \Box \neg \Box \beta$             \\ \hline
            $\vt$                                                   & $\beta, \Diamond \beta, \Diamond \neg \beta$             & $\Box^{\mKBfourfive} (\vf) = \{ \vF \}$                                     & $\neg \Box \beta \land \Box\neg\Box \beta \land \Diamond \neg \Box \beta$ \\
\hline
            \end{tabular}
  
        }
    \caption{$\mKBfourfive$ family. }
    \end{subtable}
    \begin{subtable}[t]{.47\textwidth}
    \resizebox{\textwidth}{!}{
        %
    }
    \caption{$\mKT\star$ family.}
    \end{subtable}
    \caption{Showing that $v_{\Delta}^{\mK}(\cdot)$ is a level valuation in the families $\mKBfourfive$ and $\mKT\star$.\label{tab:adq2}}
\end{table}

    \begin{table}[h!]\centering
        \resizebox{\textwidth}{!}{
        %
}
\vspace{0.3cm}
    \caption{Showing that $v_{\Delta}^{\mK}(\cdot)$ is a level valuation in the $\mKD\star$ family.\label{tab:adq3}}
    \end{table}

\begin{table}[h!]\centering
    \resizebox{.97\textwidth}{!}{
        \begin{tabular}{|l|l|l|l|}
            \hline
            \multicolumn{1}{|p{1cm}|}{$v_\Delta(\beta)$, $v_\Delta(\gamma)$} & \multicolumn{1}{c|}{Hyp: $\Delta \vdash^{\mK*}$}      & \multicolumn{1}{c|}{$\tilde \to $} & \multicolumn{1}{c|}{Goal: $\Delta \vdash^{\mK*}$}                                                   \\ \hline
            \multirow{2}{*}{$\vF, \vF$}                                & $\neg \beta, \Box \neg \beta, \Diamond \neg \beta$    & \multirow{2}{*}{$\vT$}                                                & \multirow{2}{*}{$\beta \to \gamma \land \Diamond (\beta \to \gamma) \land \Box (\beta \to \gamma)$} \\ \cline{2-2}
                                                                       & $\neg \gamma, \Box \neg \gamma, \Diamond \neg \gamma$ &                                                                       &                                                                                                     \\ \hline
            \multirow{2}{*}{$\vF, \vf$}                                & $\neg \beta, \Box \neg \beta, \Diamond \neg \beta$    & \multirow{2}{*}{$\vT$}                                                & \multirow{2}{*}{$\beta \to \gamma \land \Diamond (\beta \to \gamma) \land \Box (\beta \to \gamma)$} \\ \cline{2-2}
                                                                       & $\neg \gamma, \Diamond \gamma, \Diamond \neg \gamma$  &                                                                       &                                                                                                     \\ \hline
            \multirow{2}{*}{$\vF, \vfff$}                              & $\neg \beta, \Box \neg \beta, \Diamond \neg \beta$    & \multirow{2}{*}{$\vT$}                                                & \multirow{2}{*}{$\beta \to \gamma \land \Diamond (\beta \to \gamma) \land \Box (\beta \to \gamma)$} \\
\hline
            \end{tabular}

        \begin{tabular}{|l|l|l|l|}
            \hline
            \multicolumn{1}{|p{1cm}|}{$v_\Delta(\beta)$, $v_\Delta(\gamma)$} & \multicolumn{1}{c|}{Hyp: $\Delta \vdash^{\mK*}$}      & \multicolumn{1}{c|}{$\tilde \to $} & \multicolumn{1}{c|}{Goal: $\Delta \vdash^{\mK*}$}                                                   \\ \hline
            \multirow{2}{*}{$\vf, \vF$}                                & $\neg \beta, \Diamond \beta, \Diamond \neg \beta$     & \multirow{2}{*}{$\vt$}                                                & \multirow{2}{*}{$\beta \to \gamma \land \Diamond (\beta \to \gamma) \land \Diamond \neg (\beta \to \gamma)$} \\ \cline{2-2}
                                                                       & $\neg \gamma, \Box \neg \gamma, \Diamond \neg \gamma$ &                                                                       &                                                                                                              \\ \hline
            \multirow{2}{*}{$\vf, \vf$}                                & $\neg \beta, \Diamond \beta, \Diamond \neg \beta$     & \multirow{2}{*}{$\vT, \vt$}                                           & $(\beta \to \gamma \land \Diamond (\beta \to \gamma) \land \Box (\beta \to \gamma)) \lor$                    \\ \cline{2-2}
                                                                       & $\neg \gamma, \Diamond \gamma, \Diamond \neg \gamma$  &                                                                       & $(\beta \to \gamma \land \Diamond (\beta \to \gamma) \land \Diamond \neg (\beta \to \gamma))$                \\ \hline
            \multirow{2}{*}{$\vf, \vfff$}                              & $\neg \beta, \Diamond \beta, \Diamond \neg \beta$     & \multirow{2}{*}{$\vT$}                                                & \multirow{2}{*}{$\beta \to \gamma, \Diamond (\beta \to \gamma), \Box (\beta \to \gamma)$}                    \\
\hline
            \end{tabular}

    }

    \resizebox{.97\textwidth}{!}{

    \begin{tabular}{|l|l|l|l|}
        \hline
        $v_\Delta(\beta)$, $v_\Delta(\gamma)$ & \multicolumn{1}{c|}{Hyp: $\Delta \vdash^{\mK*}$}      & \multicolumn{1}{c|}{$\tilde \to $} & \multicolumn{1}{c|}{Goal: $\Delta \vdash^{\mK*}$}                                                             \\ \hline
        \multirow{2}{*}{$\vfff, \vF$}         & $\neg \beta, \Box \beta, \Diamond \beta$              & \multirow{2}{*}{$\vttt$}           & \multirow{2}{*}{$\beta \to \gamma \land \Diamond \neg (\beta \to \gamma) \land \Box \neg (\beta \to \gamma)$} \\ \cline{2-2}
                                              & $\neg \gamma, \Box \neg \gamma, \Diamond \neg \gamma$ &                                    &                                                                                                               \\ \hline
        \multirow{2}{*}{$\vfff, \vf$}         & $\neg \beta, \Box \beta, \Diamond \beta$              & \multirow{2}{*}{$\vt$}             & \multirow{2}{*}{$\beta \to \gamma \land \Diamond (\beta \to \gamma) \land \Diamond \neg (\beta \to \gamma)$}  \\ \cline{2-2}
                                              & $\neg \gamma, \Diamond \gamma, \Diamond \neg \gamma$  &                                    &                                                                                                               \\ \hline
        \multirow{2}{*}{$\vfff, \vfff$}       & $\neg \beta, \Box \beta, \Diamond \beta$              & \multirow{2}{*}{$\vT$}             & \multirow{2}{*}{$\beta \to \gamma \land \Diamond (\beta \to \gamma) \land \Box (\beta \to \gamma)$}           \\ \cline{2-2}
                                              & $\neg \gamma, \Box \gamma, \Diamond \gamma$           &                                    &                                                                                                               \\
\hline
        \end{tabular}

        \begin{tabular}{|l|l|l|l|}
            \hline
            \multicolumn{1}{|p{1cm}|}{$v_\Delta(\beta)$, $v_\Delta(\gamma)$} & \multicolumn{1}{c|}{Hyp: $\Delta \vdash^{\mK*}$}      & \multicolumn{1}{c|}{$\tilde \to $} & \multicolumn{1}{c|}{Goal: $\Delta \vdash^{\mK*}$}                                                   \\ \hline
            \multirow{2}{*}{$\vttt, \vF$}                              & $\beta, \Box \neg \beta, \Diamond \neg \beta$         & \multirow{2}{*}{$\vfff$}                                              & \multirow{2}{*}{$\neg (\beta \to \gamma) \land \Diamond (\beta \to \gamma) \land \Box (\beta \to \gamma)$} \\ \cline{2-2}
                                                                       & $\neg \gamma, \Box \neg \gamma, \Diamond \neg \gamma$ &                                                                       &                                                                                                            \\ \hline
            \multirow{2}{*}{$\vttt, \vf$}                              & $\beta, \Box \neg \beta, \Diamond \neg \beta$         & \multirow{2}{*}{$\vfff$}                                              & \multirow{2}{*}{$\neg (\beta \to \gamma) \land \Diamond (\beta \to \gamma) \land \Box (\beta \to \gamma)$} \\ \cline{2-2}
                                                                       & $\neg \gamma, \Diamond \gamma, \Diamond \neg \gamma$  &                                                                       &                                                                                                            \\ \hline
            \multirow{2}{*}{$\vttt, \vfff$}                            & $\beta, \Box \neg \beta, \Diamond \neg \beta$         & \multirow{2}{*}{$\vfff$}                                              & \multirow{2}{*}{$\neg (\beta \to \gamma) \land \Diamond (\beta \to \gamma) \land \Box (\beta \to \gamma)$} \\
\hline
            \end{tabular}

    }

    \resizebox{.97\textwidth}{!}{

    \begin{tabular}{|l|l|l|l|}
        \hline
            \multicolumn{1}{|p{1cm}|}{$v_\Delta(\beta)$, $v_\Delta(\gamma)$} & \multicolumn{1}{c|}{Hyp: $\Delta \vdash^{\mK*}$}      & \multicolumn{1}{c|}{$\tilde \to $} & \multicolumn{1}{c|}{Goal: $\Delta \vdash^{\mK*}$}                                                   \\ \hline
        \multirow{2}{*}{$\vt, \vF$}                                & $\beta, \Diamond \beta, \Diamond \neg \beta$          & \multirow{2}{*}{$\vf$}                                                & \multirow{2}{*}{$\neg (\beta \to \gamma) \land \Diamond (\beta \to \gamma) \land \Diamond \neg (\beta \to \gamma)$} \\ \cline{2-2}
                                                                   & $\neg \gamma, \Box \neg \gamma, \Diamond \neg \gamma$ &                                                                       &                                                                                                                     \\ \hline
        \multirow{2}{*}{$\vt, \vf$}                                & $\beta, \Diamond \beta, \Diamond \neg \beta$          & \multirow{2}{*}{$\vf, \vfff$}                                         & $(\neg (\beta \to \gamma) \land \Diamond (\beta \to \gamma) \land \Diamond \neg (\beta \to \gamma)) \lor$           \\ \cline{2-2}
                                                                   & $\neg \gamma, \Diamond \gamma, \Diamond \neg \gamma$  &                                                                       & $(\neg (\beta \to \gamma) \land \Box (\beta \to \gamma) \land \Diamond (\beta \to \gamma))$                         \\ \hline
        \multirow{2}{*}{$\vt, \vfff$}                              & $\beta, \Diamond \beta, \Diamond \neg \beta$          & \multirow{2}{*}{$\vfff$}                                              & \multirow{2}{*}{$(\neg (\beta \to \gamma) \land \Box (\beta \to \gamma) \land \Diamond (\beta \to \gamma))$}        \\
\hline
        \end{tabular}

        \begin{tabular}{|l|l|l|l|}
            \hline
            \multicolumn{1}{|p{1cm}|}{$v_\Delta(\beta)$, $v_\Delta(\gamma)$} & \multicolumn{1}{c|}{Hyp: $\Delta \vdash^{\mK*}$}      & \multicolumn{1}{c|}{$\tilde \to $} & \multicolumn{1}{c|}{Goal: $\Delta \vdash^{\mK*}$}                                                   \\ \hline
            \multirow{2}{*}{$\vT, \vF$}                                & $\beta, \Diamond \beta, \Diamond \neg \beta$          & \multirow{2}{*}{$\vF$}                                                & \multirow{2}{*}{$\neg (\beta \to \gamma) \land \Diamond \neg (\beta \to \gamma) \land \Box \neg (\beta \to \gamma)$} \\ \cline{2-2}
                                                                       & $\neg \gamma, \Box \neg \gamma, \Diamond \neg \gamma$ &                                                                       &                                                                                                                      \\ \hline
            \multirow{2}{*}{$\vT, \vf$}                                & $\beta, \Diamond \beta, \Diamond \neg \beta$          & \multirow{2}{*}{$\vf$}                                                & \multirow{2}{*}{$\neg (\beta \to \gamma) \land \Diamond \neg (\beta \to \gamma) \land \Diamond (\beta \to \gamma)$}  \\ \cline{2-2}
                                                                       & $\neg \gamma, \Diamond \gamma, \Diamond \neg \gamma$  &                                                                       &                                                                                                                      \\ \hline
            \multirow{2}{*}{$\vT, \vfff$}                              & $\beta, \Diamond \beta, \Diamond \neg \beta$          & \multirow{2}{*}{$\vfff$}                                              & \multirow{2}{*}{$\neg (\beta \to \gamma) \land \Diamond (\beta \to \gamma) \land \Box (\beta \to \gamma)$}           \\
\hline
            \end{tabular}

    }

    \resizebox{.5\textwidth}{!}{
    \begin{tabular}{|l|l|l|l|}
        \hline
            \multicolumn{1}{|p{1cm}|}{$v_\Delta(\beta)$, $v_\Delta(\gamma)$} & \multicolumn{1}{c|}{Hyp: $\Delta \vdash^{\mK*}$}      & \multicolumn{1}{c|}{$\tilde \to $} & \multicolumn{1}{c|}{Goal: $\Delta \vdash^{\mK*}$}                                                   \\ \hline
        \multirow{2}{*}{$\vff, \vff$}                              & $\neg \beta, \Box \neg \beta, \Box \beta$        & \multirow{2}{*}{$\vtt$}                                               & \multirow{2}{*}{$\beta \to \gamma \land \Box \neg (\beta \to \gamma) \land \Box \neg (\beta \to \gamma)$}        \\ \cline{2-2}
                                                                   & $\neg \gamma, \Box \neg \gamma, \Box \gamma$     &                                                                       &                                                                                                                  \\ \hline
        \multirow{2}{*}{$\vff, \vtt$}                              & $\neg \beta, \Box \neg \beta, \Box \beta$        & \multirow{2}{*}{$\vtt$}                                               & \multirow{2}{*}{$\beta \to \gamma \land \Box \neg (\beta \to \gamma) \land \Box \neg (\beta \to \gamma)$}        \\ \cline{2-2}
                                                                   & $\gamma, \Box \neg \gamma, \Box \gamma$          &                                                                       &                                                                                                                  \\ \hline
        \multirow{2}{*}{$\vtt, \vff$}                              & $\beta, \Box \neg \beta, \Box \beta$             & \multirow{2}{*}{$\vff$}                                               & \multirow{2}{*}{$\neg (\beta \to \gamma) \land \Box \neg (\beta \to \gamma) \land \Box \neg (\beta \to \gamma)$} \\
\hline
        \end{tabular}
    }
    \vspace{0.3cm}
    \caption{Showing that $v_{\Delta}^{\mK}(\cdot)$ is a level valuation. The cases for implication for all the logics. In $\mKT*$, and $\mKBfourfive$ $\tilde \to (\vt, \vf) = \{ \vf \}$, given that $\vfff$ is not a valid truth-value on those families. However, it can be checked that the same goal is satisfied. The other cases are the same for all the logics, changing only the domain of $\tilde \to$. \label{tab:cfunc-imp}}
\end{table}

    \subsection{Proofs in \Cref{sec:truth}}\label{app:sec-t}

The proof of \Cref{th:analy} requires to show that a ``total'' partial valuation (defined on all the formulas)
is indeed a level valuation. 

\begin{lemma}[Analyticity of total functions]\label{lemma:totalf}
    For every partial level valuation $v$, if $v$ is total
   % if $Dom(v) = For(\Omega)$, 
    then $v \in \lv{n}(\mathcal{M})$ for all $n$.
\end{lemma}
\begin{proof}
     By induction  on $n$. 

    \textbf{($n = 0$)} Clearly $v$ is a valuation function and hence $v \in \lv{0}(\M)$.

    \textbf{(IH)} For every $0 \leq i \leq n$, $v \in \lv{i}(\M)$. 

    \textbf{Inductive step} Assume that $v \not\in \lv{n+1}(\M)$. 
Then, there is some formula $\beta$ such that 
$\vDash^{\mathcal{L}_{n}} \beta$, but $v(\beta) \not\in \{ \vT, \vtt \}$. 
Hence, either $v (\beta) \not\in \mathcal{D}$ or $v(\beta) \in \{ \vt, \vttt \}$. 
The first possibility clearly contradicts the hypothesis. Now, assume that 
$v(\beta) = \vt$. In this case, $v(\beta) \in \sPN$. 
Moreover, by (IH), $v \in \lv{n}(\M)$. Hence, 
there is a  partial level valuation $ w \in \lv{n}(\M)$ such that 
$w(\beta) \not\in \mathcal{D}$, which is a contradiction. The case  $v(\beta) = \vttt$ is similar.
\qed
\end{proof}

%\begin{proof}[of \Cref{lem:sound-LV} ]
%    By induction on the derivation of $\Gamma \vdash^{L}\alpha$. In the case of $\mpo$, 
%    \Cref{lemma:extlemma} is used to show that a valuation 
%    considered in the fact that  $\Gamma \vDash \alpha$ must necessarily 
%    be considered in $\Gamma \vDash \alpha \to \beta$.
%
%    \begin{description}
%        \item[Base case] In this case, $\alpha$ is an instance of an axiom or $\alpha \in \Gamma$. First case is proved by previous lemmas, for each modal system. Second case follows from the definition of consequence. 
%        \item[MP] Assume that there is some proof with length $n$ of $\Gamma \vdash^{\mK*} \beta \to \psi$ and $\Gamma \vdash^{\mK*} \beta$. By (MP), $\Gamma \vdash^{\mK*} \psi$ (at step $n+1$). Now, by inductive hypothesis, for every model $\langle v, R \rangle$ satisfying $\Gamma$, $v(\beta \to \psi) \in \mathcal{D}$ and $v(\beta) \in \mathcal{D}$. Therefore, by the nmatrix of implication, $v(\psi) \in \mathcal{D}$.
%        \item[RN] Assume that there is some proof with length $n$ of $\vdash^{\mK*} \beta$. By (RN), $\vdash^{\mK*} \Box \beta$ (at step $n+1$). Now, by inductive hypothesis, $v(\beta) \in \mathcal{D}$ for every $\mK*$ model $\langle v, R \rangle$. If $v(\beta) \in \{ \vT, \vtt \}$, then $v(\Box \beta) \in \mathcal{D}$. Now, if $v(\beta) \in \{ \vt, \vttt \}$, then, by the definition of $\mK*$ model, there is some partial level valuation such that $w(\beta) \not\in \mathcal{D}$, which contradicts the hypothesis.
%    \end{description}
%\end{proof}

The following example illustrates the fact that a model $\langle \Pi, R  \rangle$ does not necessarily 
satisfy the expected frame conditions. However, due to \Cref{th:frame},
$R$ can always be extended with additional pairs to meet those conditions. 
\begin{example}\label{ex:frames}
    Consider the following valuations in $\mSfour$. \\

\begin{tabular}{|c|c|c|c|c|c|}
\hline
& $~\ldots$ \; &$\alpha$ & $\beta$ &  \ldots\\ \hline
$\vv_1$ & $\ldots$& $\vf$  &  $\vt$ & \ldots  \\ \hline
$\vv_2$ & $\ldots$  & $\vt$  &  $\vF$ & \ldots \\ \hline
$\vv_3$ & $\ldots$   & $\vT$   & $\vF$ &\ldots \\ %\hline
$\vdots$    &    &    & & \\ \hline
\end{tabular}
\\

The relation $R=\{(v_1,v_2), (v_2,v_3)\} \cup \{(\vv_i,\vv_i)\mid i\in \{1,2,3\}\}$
is enough to show that $\langle \Pi, R\rangle$, where $\Pi=\{v_i\mid i\in \{1,2,3\}\}$, 
is indeed a model (whenever $\relcond$ does not fail due to the values in the columns
not shown). Hence, $R$ is not transitive. However, we can add to $R$ the pair $(v_1,v_3)$
(indeed, the $\sP$ and $\sPN$ values in $\vv_1$ can be also supported by $\vv_3$)
resulting in a transitive relation. 
\end{example}

\begin{proof}[of \Cref{lemma:extlemma}] 
    We show below how to extend a partial valuation with a formula $\Box \alpha$
    given that $\alpha$ is in the domain of the valuation. 
    We consider the non-deterministic cases of the multifunction $\tilde\Box$ in each logic, 
    and provide a criterion to select on of the possible values. 
\begin{description}
    \item[$\mKTB$] If $\exists v'$ such that $vRv'$ and $v'(\alpha) \in \{ \vT, \vt \}$ and 
        $v'(\alpha) \in \{ \vF, \vf \}$, then extend $v$ with $\vt$. Else, extend with $\vT$.
    \item[$\mKTfour$]  If $\exists v'$ such that $vRv'$ and $v'(\alpha) = \vT$, then extend $v$ with $\vf$. Else, extend with $\vF$.
    \item[$\mK, \mKD, \mKDB, \mKB$] \ 
    \begin{enumerate}
        \item If, $\forall v'$ such that $vRv'$, $v'(\alpha) \in \sN$, then extend with $\vT, \vfff$.
        \item If, $\forall v'$ such that $vRv'$, $v'(\alpha) \not\in \sN$, then extend with $\vF, \vttt$.
        \item If $\exists v', v''$ such that $vRv'$ and $vRv''$, $v'(\alpha) \in \sN$ and $v''(\alpha) \not\in \sN$, then extend with $\vf, \vt$.
    \end{enumerate}
    \item[$\mKfour, \mKDfour$] \ 
    \begin{enumerate}
        \item If, $\forall v'$ such that $vRv'$, $v'(\alpha) \in \sN$, then extend with $\vfff$.
        \item If, $\forall v'$ such that $vRv'$, $v'(\alpha) \not\in \sN$, then extend with $\vF$.
        \item If $\exists v', v''$ such that $vRv'$ and $vRv''$, $v'(\alpha) \in \sN$ and $v''(\alpha) \not\in \sN $, then extend with $\vf$.
    \end{enumerate}
    \item[$\mKfive, \mKDfive$]  If $\exists v'$ such that $vRv'$ and $v'(\alpha) \not\in \sN$, then extend with $\vttt$. Otherwise, extend with $\vT$.
\end{description}

It is worth noticing that  we always take the  same  decision to extend a valuation 
$v$ by considering the others valuations $v'$ such that $vRv'$. To
decide how to extend $v'$, we use the same criteria, based on  the frame conditions of
$\relcond$ (\Cref{th:frame}) in each case. For example, in any logic extending $\mKB$, if we
choose a lower case in $v$, we should also choose a lower case in $v'$, due to
symmetry. In any logic extending $\mKfour$, if we choose an $\sN$-value (or
$\sI$-value), then we must keep choosing an $\sN$-value (or $\sI$-value), due
to transivity.\qed
\end{proof}

\begin{proof}[of \Cref{th:analy}]
    By \Cref{lemma:extlemma}, and a suitable enumeration of $\For$, a partial valuation can be extended to 
    a valuation with domain $\For$. By \Cref{lemma:totalf}, this extended valuation is a
     level valuation. \qed
    %Let $\tilde v$ be a partial valuation. Define $\tilde v = \tilde v^0_0$. By Lemma~\ref{lemma:extlemma}, $v_\omega$ is a partial level valuation. By Lemma~\ref{lemma:totalf}, $v_\omega$ is a level valuation.
\end{proof}

\begin{proof}[of \Cref{th:co-analy}]
As noticed in the proof sketch of this result, we have to identify 
a set of formulas $\Delta$ taking a value $\valtwo$ in 
 $\vv$, and a formula $\psi_\valtwo(\alpha)$
to show that, indeed, 
the value $\val(\alpha)$  must necessarily be supported. 
The needed $\Delta$ in each case is given in 
\Cref{table:delta}. The 
needed formula  $\psi_\valtwo(\alpha)$ is obtained as follows:

\noindent For $\mKT\star, \mKBfourfive$: 

\begin{tabular}{ccc}
    $\scriptsize
\psi_\vf (A) = 
     \begin{cases}
        \vF(A) \lor \vf (A) & \mKT, \mKTfour \\
        \vF(A) \lor \vf(A) \lor \vT(A) & \mKTB \\ 
        \neg \vt(A)& \mKTBfourfive, \mKBfourfive
     \end{cases}
$ & \qquad &\\
$\scriptsize
\psi_\vt (A) = 
     \begin{cases}
        \vT(A) \lor \vt (A) & \mKT, \mKTfour \\
        \vT(A) \lor \vt (A) \lor \vF(A) & \mKTB \\
        \neg \vf(A) & \mKTBfourfive, \mKBfourfive
     \end{cases}
$
\end{tabular}

\noindent For $\mKD\star$:

\begin{tabular}{ccc}
$\scriptsize
    \psi_{\vf} (A) = 
         \begin{cases}
            \neg A & \mKD, \mKDfour \\
            \neg A \lor \vT(A) & \mKDB \\ 
            \neg A \lor \vttt(A) & \mKDfive \\
            \neg \vt (A) & \mKDfourfive
         \end{cases}
$ \qquad 
$\scriptsize
       \psi_{\vt} (A) = 
         \begin{cases}
            A & \mKD, \mKDfour \\
            A \lor \vF(A) & \mKDB \\ 
            A \lor \vfff(A) & \mKDfive \\
            \neg \vf (A) & \mKDfourfive
         \end{cases}
$\\
$\scriptsize
\psi_{\vfff} (A) = 
     \begin{cases}
        \vF(A) \lor \vf (A) \lor \vfff (A) & \mKD \\
        \vF(A) \lor \vf(A) \lor \vfff (A) \lor \vT(A) & \mKDB \\ 
        \vF(A) \lor \vf(A) \lor \vfff (A) \lor \vttt(A) & \mKDfive \\
        \neg \vT (A) & \mKDfour, \mKDfourfive
     \end{cases}
$\\
$\scriptsize
\psi_{\vttt} (A) = 
     \begin{cases}
        \vT(A) \lor \vt (A) \lor \vttt (A) & \mKD \\
        \vT(A) \lor \vt(A) \lor \vttt (A) \lor \vF(A) & \mKDB \\ 
        \vT(A) \lor \vt(A) \lor \vttt (A) \lor \vfff(A) & \mKDfive \\
        \neg \vF (A) & \mKDfour, \mKDfourfive
     \end{cases}
$
\end{tabular}

\noindent For $\mK\star$: 

\noindent\begin{tabular}{ccc}
$\scriptsize
    \psi_{\vf} (A) = 
         \begin{cases}
            \neg A & \mK, \mKfour \\
            \neg A \lor \vT(A) \lor \vtt(A) & \mKB \\ 
            \neg \vt (A) & \mKfive, \mKfourfive
         \end{cases}
         $&\quad&
$\scriptsize
    \psi_{\vt} (A) = 
         \begin{cases}
            A & \mK, \mKfour \\
            A \lor \vF(A) \lor \vff(A) & \mKB \\ 
            \neg \vf (A) & \mKfive, \mKfourfive
         \end{cases}
$\\\\
$\scriptsize
    \psi_{\vfff} (A) = 
         \begin{cases}
            \neg A & \mK \\
            \neg A \lor \vT(A) \lor \vtt(A) & \mKB \\ 
            \neg A \lor \vt(A) \lor \vttt(A) & \mKfour \\
            \neg A \lor \vtt (A) \lor \vttt (A) & \mKfive \\
            \neg \vT(A) & \mKfourfive
         \end{cases}
         $&\quad&
$\scriptsize
    \psi_{\vttt} (A) = 
         \begin{cases}
            A & \mK \\
            A \lor \vF(A) \lor \vff(A) & \mKB \\ 
            A \lor \vf(A) \lor \vfff(A) & \mKfour \\
            A \lor \vff (A) \lor \vfff (A) & \mKfive \\
            \neg \vF(A) & \mKfourfive
         \end{cases}
$\\\\
$\scriptsize
\psi_\vf (A) = 
     \begin{cases}
        \vF(A) \lor \vf (A) & \mKT, \mKTfour \\
        \vF(A) \lor \vf(A) \lor \vT(A) & \mKTB \\ 
        \neg \vt(A)& \mKTBfourfive
     \end{cases}
$
\end{tabular}

Proving the expected properties of these formulas  is just a matter
of using the respective Nmatrices. For instance, in $\mKTB$ (abusing of the notation of $\tilde\conn$ and omitting $\{\cdot\}$ in 
singletons), we show below that $\psi_\val(\alpha)[\val] \in \{\vt,\vttt\}$:

{\scriptsize
    \begin{align*}
        \vf(A)[\vf] &= \tilde\neg (\vf) \tilde\land \tilde\Diamond^{\mKTB} (\vf) \tilde\land \tilde\Diamond^{\mKTB} (\tilde\neg (\vf)) \\
        &= \vt \tilde\land \{ \vT, \vt \} \tilde\land \tilde\Diamond^{\mKTB} (\vt) \\
        &= \vt \tilde\land \{ \vT, \vt \} \tilde\land \vT \\
        &= \vt \\
        \vF(A)[\vf] &= \tilde\neg (\vf) \tilde\land \tilde\Box^{\mKTB} (\tilde\neg (\vf)) \\
        &= \vt \tilde\land \tilde\Box^{\mKTB} (\vt) \\
        &= \vt \tilde\land \{ \vF, \vf \} \\
        &= \vF \\
        \vT(A)[\vf] &= \vf \tilde\land \tilde\Box^{\mKTB} (\vf) \\
        &= \vf \tilde\land \vF \\
        &= \vF \\
     &\therefore \vF(A)[\vf] \tilde\lor \vf (A)[\vf] \tilde\lor \vT(A)[\vf] = \vt
    \end{align*} 
}
    \qed
\end{proof}

\begin{table}\centering
\begin{subtable}[b]{0.47\textwidth}
\resizebox{\textwidth}{!}{
    \begin{tabular}{|l|l|}
        \hline
        \multicolumn{1}{|c|}{\textbf{Case}} & \multicolumn{1}{c|}{\textbf{$\Delta$}}                                          \\ \hline
        $\vT \relcond \vT, \vt$             & $\Delta = \{ \beta \mid v_p(\beta) = \vT \}$                                    \\ \hline
        $\vF \relcond \vF, \vf$             & $\Delta = \{ \neg \beta \mid v_p(\beta) = \vF \}$                               \\ \hline
        $\vt \relcond \vT, \vt, \vf$        & $\Delta = \{ \Diamond \beta \mid v_p(\beta) = \vt  \}$                          \\ \hline
        $\vf \relcond \vF, \vt, \vf$        & $\Delta = \{\Diamond \neg \beta \mid v_p(\beta) = \vf \}$                       \\ \hline
        $\vT \relcond \vT$                  & $\Delta = \{ \Box \beta \mid v_p(\beta) = \vT \}$                               \\ \hline
        $\vF \relcond \vF$                  & $\Delta = \{ \Box \neg \beta \mid v_p(\beta) = \vF \}$                          \\ \hline
        $\vt \relcond \vt, \vf$             & $\Delta = \{ \Diamond \beta \land \Diamond \neg \beta \mid v_p(\beta) = \vt \}$ \\ \hline
        $\vf \relcond \vt, \vf$             & $\Delta = \{ \Diamond \beta \land \Diamond \neg \beta \mid v_p(\beta) = \vf \}$ \\ \hline
        \end{tabular}
}
		\vspace{0.3cm}
        \caption{$\Delta$ for $\mKT\star$ and $\mKBfourfive$.\label{tab:delta1}}
    \end{subtable}
\begin{subtable}[b]{0.47\textwidth}
\resizebox{\textwidth}{!}{
        \begin{tabular}{|l|l|}
        \hline
        \multicolumn{1}{|c|}{\textbf{Case}} & \multicolumn{1}{c|}{\textbf{$\Delta$}}                                   \\ \hline
        $N \relcond \mathcal{D}$            & $\Delta = \{ A \to C \mid v_p(C) \in N \}$                               \\ \hline
        $I \relcond \bar{\mathcal{D}}$      & $\Delta = \{ A \to \neg C \mid v_p(C) \in I \}$                          \\ \hline
        $\mathcal{D} \relcond P$            & $\Delta = \{ A \to \Diamond C \mid v_p(C) \in \mathcal{D} \}$            \\ \hline
        $\bar{\mathcal{D}} \relcond PN$     & $\Delta = \{ A \to \Diamond \neg C \mid v_p(C) \in \bar{\mathcal{D}} \}$ \\ \hline
        $N \relcond N$                      & $\Delta = \{ \Box C \mid v_p(C) \in N \}$                                \\ \hline
        $I \relcond I$                      & $\Delta = \{ \Box \neg C \mid v_p(C) \in I \}$                           \\ \hline
        $P \relcond P$                      & $\Delta = \{ \Diamond C \mid v_p(C) \in  P \}$                           \\ \hline
        $PN \relcond PN$                    & $\Delta = \{ \Diamond \neg C \mid v_p(C) \in  PN \}$                     \\ \hline
        \end{tabular}
}
		\vspace{0.3cm}
        \caption{$\Delta$ for $\psi_\vfff$ and $\psi_\vf$ in $\mKD\star$ and $\mK\star$. For $\psi_\vttt$ and $\psi_\vt$, replace $A$ with $\neg A$. \label{tab:delta2}}
        \end{subtable}
\caption{Needed $\Delta$ for the proof of Co-analyticity \label{table:delta}}
\end{table}

\end{document}